\documentclass[a4paper,10pt]{article}
\usepackage[T1]{fontenc} % codifica dei font in uscita
\usepackage[utf8]{inputenc} % lettere accentate da tastiera
\usepackage[italian,english]{babel} % lingua principale del documento
\usepackage{amsmath,amssymb,mathrsfs,amsthm,amscd}
\usepackage{mathtools} % per il comando rcases
\usepackage{xcolor,booktabs} % per colorare il testo
\usepackage{ragged2e} % pacchetto per giustificare il testo nell'ambiente center
\usepackage{braket} % per usare il comando \Set e per le parentesi bra e ket
\usepackage{appendix} % per i titoletti dell'appendice
\usepackage{titlesec}
\titleformat{\section}{\normalfont\Large\scshape}{\thesection.}{.3em}{\filcenter}
\titlespacing{\section}{\parindent}{3.5ex plus .1ex minus .2ex}{3.5ex plus .2ex}
\titleformat{\subsection}{\normalfont\scshape}{\thesubsection.}{.5em}{}
\titlespacing{\subsection}{\parindent}{2.5ex plus .1ex minus .2ex}{2ex plus .2ex}
\titleformat{\subsubsection}[runin]{\normalfont\normalsize}{\thesubsubsection.}{.5em}{\itshape}[.]
\titlespacing{\subsubsection}{\parindent}{1.5ex plus .1ex minus .2ex}{1.5ex plus .2ex}

\theoremstyle{definition}
\newtheorem{defn}{Definition}[subsection]
\theoremstyle{definition}
\newtheorem{rem}{Remark}[subsection]
\theoremstyle{plain}
\newtheorem{lemma}{Lemma}[subsection]
\newtheorem{prop}{Proposition}[subsection]
\newtheorem{theorem}{Theorem}[subsection]

\usepackage{makeidx}

\begin{document}

\centerline {\LARGE {\bf {\scshape Renormalization of Higher Currents}}}
\vskip 10 pt
\centerline {\LARGE {\bf {\scshape of the sine-Gordon Model in pAQFT}}}
\vskip 50 pt
\centerline {\normalsize {\scshape Fabrizio Zanello}}
\vskip 10 pt
\centerline {\normalsize {\scshape Mathematisches Institut, Georg-August Universit\"at G\"ottingen}}
\centerline {\small {\itshape Email address:} {\ttfamily fabrizio.zanello@mathematik.uni-goettingen.de}}
\vskip 200 pt
\begin{abstract}
\scriptsize
In this paper we show that the higher currents of the sine-Gordon model are super-renormalizable by power counting in the framework of pAQFT. First we obtain closed recursive formulas for the higher currents in the classical theory and introduce a suitable notion of degree for their components. We then move to the pAQFT setting and, by means of some technical results, we compute explicit formulas for the unrenormalized interacting currents. Finally, we perform what we call the piecewise renormalization of the interacting higher currents, showing that the renormalization process involves a number of steps which is bounded by the degree of the classical conserved currents.
\end{abstract}

\thispagestyle{empty}

%\tableofcontents

\newpage

\section{Introduction}

This paper is the first step of a bigger project, which aims at understanding to what extent features of the conservation laws of classical systems are preserved when considering the corresponding quantum counterpart. More specifically, the focus of this paper is on the 2-dimensional massless sine-Gordon model, both from the point of view of Classical Field Theory and of perturbative Algebraic Quantum Field Theory (pAQFT for short).

During its long history, the sine-Gordon model has been keeping exhibiting a remarkable richness of properties. As a classical relativistic non-linear scalar field theory, it was found to be an example of integrable system. This encompasses features like: existence of an infinite number of solutions to the sine-Gordon equation (see \cite{Shn}), related by the B\"acklund transformations, and existence of an infinite number of conserved higher currents, which moreover form a commutative algebra with respect to the Peierl's bracket (see \cite{DoiFin}). Particularly relevant to our purpose is their interpretation in terms of Noether's Theorem, proposed in \cite{Steu74} and \cite{Steu76}. 

As a quantum physical system, the sine-Gordon model admits a non-trivial scattering theory. In recent years, it has revealed remarkable features also in the context of pAQFT. In particular, as shown in \cite{BR18} and \cite{BRF21}, the scattering matrix of the 2-dimensional massless sine-Gordon model in Minkowski signature was explicitly constructed and its summability proved, building partially on older results in Euclidean signature (e.g. \cite{Frohlich}). These results represent the starting point of this work, which aims at investigating the renormalization properties of the higher currents in the framework of pAQFT.

We adopt the Epstein-Glaser point of view on renormalization and prove, as main result, that the components of the higher currents are super-renormalizable by power counting in pAQFT (see \cite{DF}). 

We remark that, compared to other approaches, in our setting we do not need Fock techniques and we are hence not concerned with Fock space representations issues.

We also point out that our argument follows from well-known results on scaling-degree-preserving extensions of distributions (\cite{BF}, \cite{D}, \cite{Horm}, \cite{Stein}) and on a notion of degree that we introduce based on the concrete expressions of the higher currents, which gives a bound on the number of counterterms necessary in the renormalization process. Unlike other renormalization techniques though, we do not compute the explicit counterterms.

However, we believe that the notations and technical results that we introduce along the way might represent a good foundation for further investigations of the summability and convergence properties of the renormalized interacting higher currents. For a discussion of the renormalizability, summability, conservation and other properties of the first of the higher conserved currents of the sine-Gordon model, namely its stress-energy tensor, where the counterterms are explicitly computed, we refer to \cite{CadFrob} and \cite{CadFrob2}.

The paper is organized as follows. Section \ref{sec:class_theory} is devoted to the classical theory of the sine-Gordon model. We generalize the work done in \cite{Steu76} to the sine-Gordon with coupling constant (so that the setting considered in \cite{Steu76} is recovered as a special case for the value of the coupling constant $a=1$) and moreover we obtain explicit expressions for the components of the higher conserved currents. Along the way, we also introduce a notion of degree that will reveal to be crucial in the discussion of the renormalization of the currents in pAQFT.

In Section \ref{sec:unrenorm} we prove some technical results, on the star products and on the time-ordered products of fields with specific properties, that allow us to find closed and explicit expressions for the unrenormalized time-ordered products and the retarded components of the currents.

Finally, in Section \ref{sec:renorm} we show the renormalizability of the components of the conserved currents. We do this in three steps: first we further expand the unrenormalized expressions of the retarded components to their very elementary parts, then we piecewise renormalize the elementary parts separately and in the end we show that reassembling the piecewise renormalized parts all together gives a well-defined renormalized version of the retarded components of the currents.

\section{Conserved currents in the classical theory}
\label{sec:class_theory}

In this section we explain how the higher conserved currents for the classical sine-Gordon model can be obtained. Conceptually we follow the same passages as in \cite{Steu76}. However we extend all the definitions (given there only for the standard sine-Gordon model) to the general case of the sine-Gordon model with coupling constant (in the following referred to as the general sine-Gordon model or simply as the sine-Gordon model). Hence also our results are more general and, moreover, we derive explicit recursive formulas for the quantities involved. \\

Let us start introducing some basic notions. The sine-Gordon model is a massless relativistic non-linear scalar field theory. The r\^ole of spacetime is played by the 2-dimensional Minkowski space $\mathbb{M}_2$. The configuration bundle of the theory is the trivial bundle $\mathbb{M}_2\times\mathbb{R}\longrightarrow\mathbb{M}_2$. Configurations are sections of this bundle, namely functions $\varphi\in C^\infty(\mathbb{M}_2)$. Adopting cartesian coordinates $(x^0=:t,x^1=:\vec{x})$ on $\mathbb{M}_2$, with Minkowski metric $\eta=\text{diag}(-1,1)$, the Lagrangian of the sine-Gordon model is written as:
\begin{equation*}
L(\varphi)\,dt\wedge d\vec{x}=(L_0+L_{\text{int}})\,dt\wedge d\vec{x}=\Big(\frac{1}{2}\partial_\mu\varphi\partial^\mu\varphi+\cos(a\varphi)\Big)\,dt\wedge d\vec{x},
\end{equation*}
where the parameter $a>0$ is called coupling constant. The corresponding Euler-Lagrange equation, also called sine-Gordon equation, is
\begin{equation*}
-\square\varphi+a\sin(a\varphi)=\partial^2_t\varphi -\partial^2_{\vec{x}}\varphi+a\sin(a\varphi)=0.
\end{equation*}

In the sequel we will always work in another system of coordinates, which turns out to be particularly useful in the description of the conservation laws of the sine-Gordon model. Light-cone coordinates $(\tau,\xi)$ are defined by:
\begin{equation*}
\tau=\frac{1}{2}(\vec{x}+t),\qquad\xi=\frac{1}{2}(\vec{x}-t).
\end{equation*}
The sine-Gordon Lagrangian in light-cone coordinates becomes
\begin{equation*}
L(\varphi)\,d\tau\wedge d\xi=\Big(\frac{1}{2}\varphi_\xi\varphi_\tau+\cos(a\varphi)\Big)\,d\tau\wedge d\xi,
\end{equation*}
and the sine-Gordon equation is
\begin{equation}
\label{eq:s-G eq}
\varphi_{\xi\tau}-a\sin(a\varphi)=0,
\end{equation}
where we also adopt the convention that subscripts $_\tau$ and $_\xi$ indicate partial derivation w.r.t. the corresponding coordinate.
\begin{rem}
The so-called standard sine-Gordon model, which is more often treated in the literature, for example also in \cite{Steu76}, is the special case for $a=1$.
\end{rem}

\subsection{Extended B\"acklund transformations}
\label{sec:s-Gcoupl}

Extended B\"acklund transformations are defined in \cite{Steu76} for the standard sine-Gordon model. We extend the definition to the general sine-Gordon model by introducing also a dependence on the coupling constant $a>0$. 
\begin{defn}
We say that the configuration $\varphi'\in C^\infty(\mathbb{M}_2)$ is obtained from a given configuration $\varphi\in C^\infty(\mathbb{M}_2)$ by an extended B\"acklund transformation $\hat{B}_\alpha$ of parameter $\alpha\in\mathbb{R}$, in notation $\varphi'=\hat{B}_\alpha\varphi$, if $\varphi'$ satisfies the following parametric PDE:
\begin{equation}
\label{eq:MBTa}
\frac{1}{2}(\varphi'+\varphi)_\xi=\frac{1}{\alpha}\sin\big[\frac{1}{2}a(\varphi'-\varphi)\big]
\end{equation}
\end{defn}
Assuming that $\varphi'$ admits a power series expansion in the parameter $\alpha$, we write:
\begin{equation}
\label{eq:powA}
\varphi'=\sum_{\nu=0}^\infty A_\nu[a,\varphi]\alpha^\nu,
\end{equation}
where the coefficients $A_\nu$ depend on both the coupling constant $a$ and the initial configuration $\varphi$. We can now substitute the power series expansion (\ref{eq:powA}) in equation (\ref{eq:MBTa}), using also the power series expansion of sine, and compare order by order in $\alpha$. Omitting the dependence on $a$ and $\varphi$, we obtain the following expressions for the first coefficients $A_\nu$:
\begin{equation*}
A_0=\varphi,\qquad A_1=\frac{2}{a}\varphi_\xi,\qquad A_2=\frac{2}{a^2}\varphi_{\xi\xi}.
\end{equation*}
Carrying out a detailed study of equation (\ref{eq:MBTa}), it is possible to obtain an explicit recursive formula for the higher coefficients $A_\nu$.
\begin{prop}
\label{prop:rec_formula}
For $\nu\ge 2$ the following recursive formula holds:
\begin{equation}
\label{eq:recA}
\begin{split}
&A_{\nu+1}=\frac{1}{a}A_{\nu,\xi}+ \\
&+\sum_{\beta=0}^{[\frac{\nu}{2}]-1}(-1)^\beta\big(\frac{1}{2}a\big)^{2(\beta+1)}\sum_{\substack{n_0,\dots,n_{\nu-2-2\beta}\ge 0 \\ n_0+\dots+n_{\nu-2-2\beta}=2\beta+3 \\ 1\cdot n_1+\dots +(\nu-2-2\beta)\cdot n_{\nu-2-2\beta}=\nu-2-2\beta}}\frac{A_1^{n_0}\cdots A_{\nu-1-2\beta}^{n_{\nu-2-2\beta}}}{n_0!\cdots n_{\nu-2-2\beta}!},
\end{split}
\end{equation}
where $[\frac{\nu}{2}]$ denotes the integer part of $\frac{\nu}{2}$.
\end{prop}
\begin{proof}
The proof is given in Appendix \ref{app:Back_coupling_const}.
\end{proof}
The first coefficients obtained using formula (\ref{eq:recA}) have the form:
\begin{equation}
\label{eq:A_rec_form}
\begin{split}
A_3&=\frac{2}{a^3}\varphi_{\xi\xi\xi}+\frac{1}{3a}\varphi_\xi^3, \\
A_4&=\frac{2}{a^4}\varphi_{4\xi}+\frac{2}{a^2}\varphi_\xi^2\varphi_{\xi\xi}.
\end{split}
\end{equation}
\begin{rem}
\label{rem:all poly}
We observe that from formula (\ref{eq:recA}) it follows that the coefficients $A_\nu$ are all polynomials in the derivatives of the configuration $\varphi$ with respect only to the ligh-cone coordinate $\xi$.
\end{rem}
\begin{rem}
As a consistency check, the expressions for the coefficients $A_\nu$ presented in \cite{Steu76} are recovered from our expressions setting $a=1$.
\end{rem}

We introduce here a notion that turns out to be crucial for the subsequent discussion of the renormalization of the higher currents in pAQFT.
\begin{defn}
\label{defn:deg_A}
Consider a configuration $\varphi\in C^\infty(\mathbb{M}_2)$. We assign a degree to its $k$-th derivative with respect to the light-cone coordinate $\xi$, by:
\begin{equation*}
\text{deg}(\varphi_{k\xi})=k,\qquad\forall k\in\mathbb{N}.
\end{equation*}
We extend this definition to monomials in the derivatives of $\varphi$ by additivity:
\begin{equation*}
\text{deg}(\varphi_{k_1\xi}\varphi_{k_2\xi}\dots\varphi_{k_N\xi})=k_1+k_2+\dots+k_N.
\end{equation*}
We say that a polynomial in the derivatives of $\varphi$ is homogeneous of degree $d$ if all its monomial terms have degree $d$. 
\end{defn}
\begin{prop}
\label{prop:deg_A}
For every $\nu\ge 0$, the coefficient $A_\nu$ is homogeneous of degree equal to $\nu$.
\end{prop}
\begin{proof}
The claim is trivial for $A_0=\varphi$, $A_1=\frac{2}{a}\varphi_\xi$ and $A_2=\frac{2}{a^2}\varphi_{\xi\xi}$. For $\nu\ge 3$ we proceed by induction. For $\nu=3$, using formulas (\ref{eq:A_rec_form}), we have
\begin{equation*}
\text{deg}(A_3)=\text{deg}\Big(\frac{2}{a^3}\varphi_{\xi\xi\xi}+\frac{1}{3a}\varphi_\xi^3\Big)=3.
\end{equation*}
Now suppose the claim is true for $\nu\le N$. The coefficient $A_{N+1}$ is given by formula (\ref{eq:recA}). The first term is $\frac{1}{a}A_{N,\xi}$ which, due to the additional derivative w.r.t. $\xi$, has degree $N+1$. The other terms are given by products
\begin{equation*}
A_1^{n_0}\cdots A_{N-1-2\beta}^{n_{N-2-2\beta}},
\end{equation*}
with the conditions on the indexes $n_0,\dots,n_{N-2-2\beta}$:
\begin{equation*}
\begin{split}
n_0+\dots+n_{N-2-2\beta}&=2\beta+3 \\
1\cdot n_1+\dots +(N-2-2\beta)\cdot n_{N-2-2\beta}&=N-2-2\beta
\end{split}
\end{equation*} 
From these conditions and additivity of the degree, it follows that
\begin{equation*}
\begin{split}
\text{deg}\big(A_1^{n_0}\cdots A_{N-1-2\beta}^{n_{N-2-2\beta}}\big)&=1\cdot n_0+\dots+(N-1-2\beta)\cdot n_{N-2-2\beta}= \\
&=n_0+n_1+\dots+n_{N-2-2\beta}+ \\
&\quad+n_1+\dots+(N-2-2\beta)\cdot n_{N-2-2\beta}= \\
&=2\beta+3+N-2-2\beta= \\
&=N+1.
\end{split}
\end{equation*}
\end{proof}

\subsection{The higher conserved currents}
\label{sec:higher cons currents}

First we restrict on-shell, namely we assume that $\varphi$ is a solution of the sine-Gordon equation (\ref{eq:s-G eq}). Then, retracing the passages of \cite{Steu76} and adapting them to our more general case (in particular using our definition of extended B\"acklund transformations $\hat{B}_\alpha$), it is easy to check that the one-parameter family of $1$-forms $s^{(\alpha)}=-s^{(\alpha)}_1d\tau+s^{(\alpha)}_2d\xi$, with components
\begin{subequations}
\begin{align}
&s^{(\alpha)}_1=\cos\big[\frac{1}{2}a(\varphi+\hat{B}_{-\alpha}\varphi)\big]+\cos\big[\frac{1}{2}a(\varphi+\hat{B}_{\alpha}\varphi)\big] \label{eq:s1} \\
&s^{(\alpha)}_2=\frac{1}{\alpha^2}\big\{2-\cos\big[\frac{1}{2}a(\varphi-\hat{B}_{-\alpha}\varphi)\big]-\cos\big[\frac{1}{2}a(\varphi-\hat{B}_{\alpha}\varphi)\big]\big\}, \label{eq:s2}
\end{align}
\end{subequations}
form a family of conserved currents, namely $\forall\alpha\in\mathbb{R}$ they satisfy the equation
\begin{equation*}
\text{div}(s^{(\alpha)})=\partial_\xi s^{(\alpha)}_1+\partial_\tau s^{(\alpha)}_2=0.
\end{equation*}

Using formula (\ref{eq:powA}) to write $\hat{B}_{\pm\alpha}\varphi$ and the power series expansion of cosine, we can expand also $s^{(\alpha)}_1$ and $s^{(\alpha)}_2$ as power series in $\alpha$. Since formulas (\ref{eq:s1}) and (\ref{eq:s2}) are symmetric in $\alpha$, only even powers will appear. We denote the results of the power series expansions by:
\begin{equation*}
s^{(\alpha)}_1=\sum_{N=0}^\infty s_1^N\alpha^{2N},\qquad s^{(\alpha)}_2=\sum_{N=0}^\infty s_2^N\alpha^{2N}.
\end{equation*}
For every order in $\alpha$ a conserved current is obtained, which we denote by
\begin{equation*}
s^N=-s^N_1d\tau+s^N_2d\xi.
\end{equation*}
\begin{prop}
\label{prop:comp_cons_curr}
The components $s^N_1$ and $s^N_2$ of the conserved currents have the following form:
\begin{equation}
\label{eq:s1_final}
\begin{split}
s_1^N=&\cos(a\varphi)\Big[2\sum_{\beta=1}^{N}(-1)^\beta\Big(\frac{1}{2}a\Big)^{2\beta}\sum_{\substack{n_1,\dots,n_{2N}\ge 0 \\ n_1+\dots+n_{2N}=2\beta \\ 1\cdot n_1+\dots +2N\cdot n_{2N}=2N}}\frac{A_1^{n_1}\cdots A_{2N}^{n_{2N}}}{n_1!\cdots n_{2N}!}\Big]+ \\
&+\sin(a\varphi)\Big[2\sum_{\beta=0}^{N-1}(-1)^{\beta+1}\Big(\frac{1}{2}a\Big)^{2\beta+1}\sum_{\substack{n_1,\dots,n_{2N}\ge 0 \\ n_1+\dots+n_{2N}=2\beta+1 \\ 1\cdot n_1+\dots +2N\cdot n_{2N}=2N}}\frac{A_1^{n_1}\cdots A_{2N}^{n_{2N}}}{n_1!\cdots n_{2N}!}\Big],
\end{split}
\end{equation}
where the coefficient of $\sin(a\varphi)$ is defined only for $N\ge 1$, and
\begin{equation}
\label{eq:s2_final}
s_2^N=2\sum_{\mu=0}^N(-1)^\mu\Big(\frac{1}{2}a\Big)^{2(\mu+1)}\sum_{\substack{n_0,\dots,n_{2(N-\mu)}\ge 0 \\ n_0+\dots+n_{2(N-\mu)}=2(\mu+1) \\ 1\cdot n_1+\dots +2(N-\mu)\cdot n_{2(N-\mu)}=2(N-\mu)}}\frac{A_1^{n_0}\cdots A_{2(N-\mu)+1}^{n_{2(N-\mu)}}}{n_0!\cdots n_{2(N-\mu)}!}.
\end{equation}
\end{prop}
\begin{proof} 
The proof is given in Appendix \ref{app:proof_prop_cons_curr}.
\end{proof}
The expressions of the first components, from formulas (\ref{eq:s1_final}) and (\ref{eq:s2_final}), are:
\begin{equation*}
\begin{cases}
s_1^0=2\cos(a\varphi), \\
s_2^0=\varphi_\xi^2,
\end{cases}
\qquad
\begin{cases}
s_1^1=-\varphi_\xi^2\cos(a\varphi)-\frac{2}{a}\varphi_{\xi\xi}\sin(a\varphi), \\
s_2^1=\frac{1}{4}\varphi_\xi^4+\frac{2}{a^2}\varphi_\xi\varphi_{\xi\xi\xi}+\frac{1}{a^2}\varphi_{\xi\xi}^2.
\end{cases}
\end{equation*}

\begin{rem}
\label{rem:comp poly}
Considering Remark \ref{rem:all poly} together with formulas (\ref{eq:s1_final}) and (\ref{eq:s2_final}), it follows that the coefficients of $\cos(a\varphi)$ and of $\sin(a\varphi)$ in the expression of $s_1^N$, and the second components $s_2^N$ are all polynomials in the derivatives of the configuration $\varphi$ with respect to the coordinate $\xi$.
\end{rem}

\begin{rem}
Again, as a consistency check, we have that we recover the expressions for the components of the conserved currents given in \cite{Steu76} from our expressions setting $a=1$.
\end{rem}

To conclude this Section, we study the properties of the degree of the components of the higher conserved currents.
\begin{prop}
\label{prop:deg_cons_curr}
Assign by convention degree equal to $0$ to $\cos(a\varphi)$ and $\sin(a\varphi)$. Then we have that:
\begin{itemize}
\item The first component $s_1^N$ of the conserved current $s^N$ is homogeneous of degree equal to $2N$.
\item The second component $s_2^N$ of the conserved current $s^N$ is homogeneous of degree equal to $2(N+1)$.
\end{itemize}
\end{prop}
\begin{proof}
The first claim follows from the observation that the coefficients of $\cos(a\varphi)$ and $\sin(a\varphi)$, in formula (\ref{eq:s1_final}), are given by sums of products of the form
\begin{equation*}
A_1^{n_1}\cdots A_{2N}^{n_{2N}},
\end{equation*}
with the condition $n_1+\dots+2N\cdot n_{2N}=2N$. All these products have degree
\begin{equation*}
\text{deg}(A_1^{n_1}\dots A_{2N}^{n_{2N}})=1\cdot n_1+\dots+2N\cdot n_{2N}=2N.
\end{equation*}

As for the second claim, from formula (\ref{eq:s2_final}) we have that $s_2^N$ is given by a finite sum of products of the form
\begin{equation*}
A_1^{n_0}\cdots A_{2(N-\mu)+1}^{n_{2(N-\mu)}},
\end{equation*}
with the conditions 
\begin{equation*}
\begin{split}
n_0+\dots+n_{2(N-\mu)}&=2(\mu+1) \\
1\cdot n_1+\dots 2(N-\mu)\cdot n_{2(N-\mu)}&=2(N-\mu).
\end{split}
\end{equation*}
The degree of each one of these products is
\begin{equation*}
\begin{split}
\text{deg}\big(A_1^{n_0}\dots A_{2(N-\mu)+1}^{n_{2(N-\mu)}}\big)&=n_0+\dots+(2(N-\mu)+1)n_{2(N-\mu)}= \\
&=n_0+n_1+\dots+n_{2(N-\mu)}+ \\
&\quad +n_1+\dots+2(N-\mu)n_{2(N-\mu)}= \\
&=2(\mu+1)+2(N-\mu)= \\
&=2(N+1).
\end{split}
\end{equation*}
\end{proof}

\section{Unrenormalized expressions for the interacting higher currents in pAQFT}
\label{sec:unrenorm}

Before discussing the technical details, we recall some basic notions regarding the framework of pAQFT. In particular, we restrict to the specific setting of the sine-Gordon model (for more extensive and general treatments, we refer for example to \cite{BF}, \cite{D}, \cite{HollWald01}, \cite{R}).

Fields, also called observables, are described by a class of smooth functionals $F\colon C^\infty(\mathbb{M}_2)\rightarrow \mathbb{C}$, called microcausal functionals and denoted by $\mathcal{F}_{\mu c}$. More generally, in adherence with the perturbative approach, one considers formal power series in $\hbar$ with coefficients in microcausal functionals, denoted by $\mathcal{F}_{\mu c}[[\hbar]]$.

The space of fields is endowed with a non commutative product
\begin{equation*}
F\star G=\sum_{n=0}^\infty\frac{\hbar}{n!}\braket{\Big(\frac{\delta^n}{\delta\varphi^n}F\Big),\,(W)^n\,\frac{\delta^n}{\delta\varphi^n}G\,},\qquad F,G\in\mathcal{F}_{\mu c}[[\hbar]],
\end{equation*}
where $\frac{\delta^n}{\delta\varphi^n}F,\frac{\delta^n}{\delta\varphi^n}G$ denote the $n$-th functional derivatives of the fields $F$ and $G$ and the bidistribution $W$ is the Wightman two-point function. This product defines on $\mathcal{F}_{\mu c}[[\hbar]]$ the structure of a Poisson $\ast$-algebra, where the Poisson bracket is given by the commutator with respect to the star product and the involution $^\ast$ by complex conjugation. The resulting Poisson $\ast$-algebra $(\mathcal{F}_{\mu c}[[\hbar]],\star,[\cdot,\cdot]_\star,^\ast)$ is called the algebra of free fields.

While the algebra of free fields represents the model algebra for observables, the physical concept of evolution is encoded by the notion of interacting fields. Roughly speaking, considering interacting fields represents the quantum equivalent of the classical restriction to on-shell fields, namely observables evaluated only on configurations that are solutions of the Euler-Lagrange equations of the theory. This is precisely the case for the higher currents of the sine-Gordon model, which are conserved only when evaluated on a configuration that is a solution of the sine-Gordon equation.

In order to construct the interacting fields, first the time-ordered product is defined as the following commutative product of fields $F_1,F_2\in\mathcal{F}_{\mu c}[[\hbar]]$:
\begin{equation*}
T_2\big(F_1\otimes F_2\big)=F_1\star_FF_2=\sum_{n=0}^\infty\frac{\hbar}{n!}\braket{\Big(\frac{\delta^n}{\delta\varphi^n}F_1\Big),\,(\Delta^F)^n\,\frac{\delta^n}{\delta\varphi^n}F_2\,},
\end{equation*}
where $\Delta^F$ is the unique Feynman propagator (we remark that in general, on curved spacetimes, the Feynman propagator is not unique). The time-ordered product of order $l$ is then:
\begin{equation*}
T_l\big(F_1\otimes\dots\otimes F_l\big):=F_1\star_F\dots\star_FF_l,\qquad F_1,\dots,F_l\in\mathcal{F}_{\mu c}[[\hbar]].
\end{equation*}

In the formula for the time-ordered product, the Feynman propagator is intended as a symmetric bidistribution defined on $\mathbb{M}_2^2$. Its wavefront set is such that the product of Feynman propagators can be defined using H\"ormander's sufficient criterion (\cite{Horm}) only outside of the diagonal. This implies that the time-ordered products $T_l\big(F_1(x_1)\otimes\dots\otimes F_l(x_l)\big)$, seen as observable-valued distributions, can be defined by H\"ormander's sufficient criterion only on a subset of $\mathbb{M}_2^l$ denoted by
\begin{equation}
\label{eq:M check}
\check{\mathbb{M}}_2^{l}:=\Set{(x_1,\dots,x_{l})\in\mathbb{M}_2^{l} | x_i\neq x_j\quad\forall 1\le i<j\le l}.
\end{equation}

The renormalization problem in pAQFT is the problem of extending these products to distributions well-defined on the whole $\mathbb{M}_2^l$. This can be done by combining a study of the properties of their wavefront set with the notion of Steinmann scaling degree (\cite{Stein}). The extension process is not always unique. The ambiguities are represented by the possibility to add derivatives of Dirac deltas up to a certain order. This renormalization framework goes under the name of Epstein-Glaser renormalization. We also remark that the pAQFT setting is defined not only on Minkowski spacetime, but applies naturally to globally hyperbolic spacetimes.

The main ingredient of the interacting picture in pAQFT is the scattering matrix $S$, defined as the generating function of the time-ordered products:
\begin{equation*}
S(F):=T\big(e_\otimes^{iF/\hbar}\big):=\sum_{n=0}^\infty\frac{1}{n!}\Big(\frac{i}{\hbar}\Big)^nT_n\big(F^{\otimes n}\big),\qquad F\in\mathcal{F}_{\mu c}[[\hbar]], 
\end{equation*}
where $T_0\big(F\big)=1$ and $T_1\big(F\big)=F$. Interacting fields $(F)_{\text{int}}$ can then be constructed by means of the Bogoliubov formula
\footnote{The idea of using Bogoliubov formula to compute the interacting components of the higher currents for the sine-Gordon model was presented by the author at the 46th LQP Workshop, Erlangen, 24\--25 June 2022. Other approaches are possible. In particular, for considerations on the stress-energy tensor in a different framework, see \cite{CadFrob} and \cite{CadFrob2}.}:
\begin{equation*}
(F)_{\text{int}}:=-i\hbar\frac{d}{d\lambda}\bigg(S(L_{\text{int}})^{\star-1}\star S(L_{\text{int}}+\lambda F)\bigg)\bigg|_{\lambda=0},
\end{equation*}
where $L_{\text{int}}$ is the interaction Lagrangian of the system under consideration. In particular for the sine-Gordon model we have $L_{\text{int}}=\cos(a\varphi)$.

The result has to be intended as a formal power series in $\hbar$ (and in the coupling constant contained in $L_{\text{int}}$) and is denoted by:
\begin{equation*}
(F)_{\text{int}}=\sum_{n=0}^\infty\frac{1}{n!}R_n\big(L_{\text{int}}^{\otimes n},F\big).
\end{equation*}
The coefficients of the series are called retarded products and are given by
\begin{equation}
\label{eq:ret field model}
R_t\big(L_{\text{int}}^{\otimes t},F\big)=\Big(\frac{i}{\hbar}\Big)^t\sum_{l=0}^t\binom{t}{l}(-1)^{t-l}\bar{T}_{t-l}\big(L_{\text{int}}^{\otimes(t-l)}\big)\star T_{l+1}\big(L_{\text{int}}^{\otimes l}\otimes F\big),
\end{equation}
where $\bar{T}_n$ are the anti-chronological products, defined as the coefficients of the inverse (in the sense of formal power series) of the scattering matrix
\begin{equation*}
S(F)^{\star-1}=\sum_{n=0}^\infty\frac{1}{n!}\Big(\frac{-i}{\hbar}\Big)^n\bar{T}_n\big(F^{\otimes n}\big).
\end{equation*}

\subsection{Time-ordered products for components $s^N_2$}
\label{sec:URtos2}

We start the discussion with the components $s^N_2$ because, as pointed out in Remark \ref{rem:comp poly}, these are polynomials in the derivatives of the configuration $\varphi$. Thanks to this fact, we can further manipulate the unrenormalized expression of the $(l+1)$-th time-ordered products occurring in formula (\ref{eq:ret field model}).
\begin{prop}
\label{prop:to-int-s2}
The unrenormalized $(l+1)$-th time-ordered product for the components $s^N_2$ can be written as a finite sum of terms:
\begin{equation}
\label{eq:to_int_s2}
\begin{split}
&T_{l+1}\Big(L_{\text{int}}^{\otimes l}\otimes s^N_2\Big)= \\
&=\sum_{j=0}^{2(N+1)}\hbar^j\sum_{\substack{j_1,\dots,j_l\ge 0 \\ j_1+\dots +j_l=j}}\frac{1}{j_1!\dots j_l!}(\Delta^F)^jT_l\Big(\frac{\delta^{j_1}}{\delta\varphi^{j_1}}L_{\text{int}}\otimes\dots\otimes\frac{\delta^{j_l}}{\delta\varphi^{j_l}}L_{\text{int}}\Big)\frac{\delta^j}{\delta\varphi^j}s^N_2. 
\end{split}
\end{equation}
\end{prop}
\begin{proof}
The result can be obtained using the formula for the time-ordered product from \cite{BR18} and splitting the exponential in the following way: 
\begin{equation*}
\begin{split}
T_{l+1}\Big(L_{\text{int}}^{\otimes l}\otimes s^N_2\Big)&=\mu\circ e^{\hbar\sum_{1\le i<j\le l+1}D^{ij}_F}\Big(L_{\text{int}}^{\otimes l}\otimes s^N_2\Big)= \\
&=\mu\circ e^{\hbar\sum_{1\le i<j\le l}D^{ij}_F}\circ e^{\hbar\sum_{i=1}^lD^{i\,l+1}_F}\Big(L_{\text{int}}^{\otimes l}\otimes s^N_2\Big),
\end{split}
\end{equation*}
where
\begin{equation*}
D^{ij}_F:=\langle\Delta^F,\,\frac{\delta}{\delta\varphi_i}\otimes\frac{\delta}{\delta\varphi_j}\rangle
\end{equation*}
and the index $i$ in $\frac{\delta}{\delta\varphi_i}$ means that the functional derivative is applied to the $i$-th term of the tensor product.

From Remark \ref{rem:comp poly} and Proposition \ref{prop:deg_cons_curr} it follows that, as a field, $s^N_2$ admits non-zero functional derivatives of order at most $2(N+1)$. Hence the second exponential series is in fact a finite sum:
\begin{equation}
\label{eq:to1}
e^{\hbar\sum_{i=1}^lD^{i\,l+1}_F}\Big(L_{\text{int}}^{\otimes l}\otimes s^N_2\Big)=\sum_{j=0}^{2(N+1)}\frac{\hbar^j}{j!}\Big(\sum_{i=1}^lD^{i\,l+1}_F\Big)^j\Big(L_{\text{int}}^{\otimes l}\otimes s^N_2\Big)
\end{equation}

Expanding the operators 
\begin{equation*}
\Big(\sum_{i=1}^lD^{i\,l+1}_F\Big)^j=\langle(\Delta^F)^j,\,\Big(\frac{\delta}{\delta\varphi_1}+\dots +\frac{\delta}{\delta\varphi_l}\Big)^j\otimes\frac{\delta^j}{\delta\varphi_{l+1}^j}\rangle,
\end{equation*}
using the multinomial formula
\begin{equation*}
\Big(\frac{\delta}{\delta\varphi_1}+\dots +\frac{\delta}{\delta\varphi_l}\Big)^j=\sum_{\substack{j_1,\dots,j_l\ge 0 \\ j_1+\dots +j_l=j}}\frac{j!}{j_1!\dots j_l!}\prod_{t=1}^l\frac{\delta^{j_t}}{\delta\varphi_t^{j_t}},
\end{equation*}
and carrying out all the computations, we finally arrived at the aimed result.
\end{proof}

\subsection{Retarded components $s^N_2$}
We can use Proposition \ref{prop:to-int-s2} to make more explicit the expression of the unrenormalized retarded product of order $t$ for components $s^N_2$:
\begin{equation}
\label{eq:rets2_0}
R_t\Big(L_{\text{int}}^{\otimes t}\otimes s^N_2\Big)=\Big(\frac{i}{\hbar}\Big)^t\sum_{l=0}^t\binom{t}{l}(-1)^{t-l}\bar{T}_{t-l}\Big(L_{\text{int}}^{\otimes (t-l)}\Big)\star T_{l+1}\Big(L_{\text{int}}^{\otimes l}\otimes s^N_2\Big).
\end{equation}

First we study the $\star$-products $\bar{T}_{t-l}\Big(L_{\text{int}}^{\otimes (t-l)}\Big)\star T_{l+1}\Big(L_{\text{int}}^{\otimes l}\otimes s^N_2\Big)$ which, according to formula (\ref{eq:to_int_s2}), are in turn composed by terms of the form
\begin{equation}
\label{eq:prod}
\bar{T}_{t-l}\Big(L_{\text{int}}^{\otimes (t-l)}\Big)\star\Big(T_l\Big(\frac{\delta^{j_1}}{\delta\varphi^{j_1}}L_{\text{int}}\otimes\dots\otimes\frac{\delta^{j_l}}{\delta\varphi^{j_l}}L_{\text{int}}\Big)\frac{\delta^j}{\delta\varphi^j}s^N_2\Big).
\end{equation}
We can prove the following slightly more general technical result that applies in particular to formula (\ref{eq:prod}).
\begin{prop}
\label{prop:prod_fields}
Consider the product of fields $A,B,C\in\mathcal{F}_{\mu c}[[\hbar]]$
\begin{equation*}
A\star(B\,C),
\end{equation*}
where $C$ is such that $\exists\,c\in\mathbb{N}$ for which $\frac{\delta^i}{\delta\varphi^i}C=0$ whenever $i>c$, while $A$ and $B$ can possibly admit non-zero functional derivatives of arbitrary order. Then the product can be written in the form:
\begin{equation}
\label{eq:starprod}
A\star(B\,C)=\sum_{k=0}^{c}\frac{\hbar^k}{k!}\bigg(\sum_{n=0}^\infty\frac{\hbar^n}{n!}\bigg(\frac{\delta^k}{\delta\varphi^k}\frac{\delta^n}{\delta\varphi^n}A\bigg)(W)^n\frac{\delta^n}{\delta\varphi^n}B\bigg)(W)^k\frac{\delta^k}{\delta\varphi^k}C.
\end{equation}
\end{prop}
\begin{proof}
The claim is obtained by explicit calculation, applying the Leibniz rule and exchanging the order of the summations as follows:
\begin{equation*}
\begin{split}
A\star(B\,C)&=\sum_{n=0}^\infty\frac{\hbar^n}{n!}\bigg(\frac{\delta^n}{\delta\varphi^n}A\bigg)(W)^n\frac{\delta^n}{\delta\varphi^n}(B\,C)= \\
&=\sum_{n=0}^\infty\frac{\hbar^n}{n!}\bigg(\frac{\delta^n}{\delta\varphi^n}A\bigg)(W)^n\bigg(\sum_{k=0}^n\binom{n}{k}\frac{\delta^{n-k}}{\delta\varphi^{n-k}}B\,\frac{\delta^k}{\delta\varphi^k}C\bigg)= \\
&=\sum_{k=0}^{c}\bigg(\sum_{n=k}^\infty\frac{\hbar^n}{n!}\binom{n}{k}\bigg(\frac{\delta^n}{\delta\varphi^n}A\bigg)(W)^n\frac{\delta^{n-k}}{\delta\varphi^{n-k}}B\bigg)\frac{\delta^k}{\delta\varphi^k}C.
\end{split}
\end{equation*}
Rescaling the index of the second summation $n\rightarrow n-k$ we conclude. 
\end{proof}
Substituting formula (\ref{eq:to_int_s2}) in the product $\bar{T}_{t-l}\Big(L_{\text{int}}^{\otimes (t-l)}\Big)\star T_{l+1}\Big(L_{\text{int}}^{\otimes l}\otimes s^N_2\Big)$ and applying Proposition \ref{prop:prod_fields}, we obtain:
\begin{equation*}
\begin{split}
&\bar{T}_{t-l}\Big(L_{\text{int}}^{\otimes (t-l)}\Big)\star T_{l+1}\Big(L_{\text{int}}^{\otimes l}\otimes s^N_2\Big)= \\
&=\sum_{j=0}^{2(N+1)}\hbar^j\sum_{\substack{j_1,\dots,j_l\ge 0 \\ j_1+\dots +j_l=j}}\frac{1}{j_1!\dots j_l!}\sum_{i=0}^{2(N+1)-j}\frac{\hbar^i}{i!}\cdot \\
&\cdot\bigg(\sum_{k=0}^\infty\frac{\hbar^k}{k!}\frac{\delta^i}{\delta\varphi^i}\frac{\delta^k}{\delta\varphi^k}\bar{T}_{t-l}\Big(L_{\text{int}}^{\otimes (t-l)}\Big)(W)^k\frac{\delta^k}{\delta\varphi^k} T_l\Big(\frac{\delta^{j_1}}{\delta\varphi^{j_1}}L_{\text{int}}\otimes\dots\otimes\frac{\delta^{j_l}}{\delta\varphi^{j_l}}L_{\text{int}}\Big)\bigg)\cdot \\
&\quad\cdot(\Delta^F)^j(W)^i\frac{\delta^i}{\delta\varphi^i}\frac{\delta^j}{\delta\varphi^j}s^N_2.
\end{split}
\end{equation*}

We can now plug this equation in formula (\ref{eq:rets2_0}) and finally arrive to the explicit expression for the unrenormalized retarded components $s^N_2$:
\begin{equation}
\label{eq:rets2}
\begin{split}
&R_t\Big(L_{\text{int}}^{\otimes t}\otimes s^N_2\Big)= \\
&=\Big(\frac{i}{\hbar}\Big)^t\sum_{l=0}^t\binom{t}{l}(-1)^{t-l}\sum_{j=0}^{2(N+1)}\sum_{i=0}^{2(N+1)-j}\frac{\hbar^j\hbar^i}{i!}\sum_{\substack{j_1,\dots,j_l\ge 0 \\ j_1+\dots +j_l=j}}\frac{1}{j_1!\dots j_l!}\cdot \\
&\cdot\bigg(\sum_{k=0}^\infty\frac{\hbar^k}{k!}\frac{\delta^i}{\delta\varphi^i}\frac{\delta^k}{\delta\varphi^k}\bar{T}_{t-l}\Big(L_{\text{int}}^{\otimes (t-l)}\Big)(W)^k\frac{\delta^k}{\delta\varphi^k} T_l\Big(\frac{\delta^{j_1}}{\delta\varphi^{j_1}}L_{\text{int}}\otimes\dots\otimes\frac{\delta^{j_l}}{\delta\varphi^{j_l}}L_{\text{int}}\Big)\bigg)\cdot \\
&\quad\cdot(\Delta^F)^j(W)^i\frac{\delta^i}{\delta\varphi^i}\frac{\delta^j}{\delta\varphi^j}s^N_2.
\end{split}
\end{equation}

\subsection{Time-ordered products for components $s^N_1$}
We know from formula (\ref{eq:s2_final}) that components $s^N_1$ are given by the sum of a homogeneous part of degree $2N$ multiplied by $\cos(a\varphi)$ and another homogeneous part of degree $2N$ multiplied by $\sin(a\varphi)$. We rename the two homogeneous parts $q^N_1$ and $r^N_1$, respectively, and write
\begin{equation}
\label{eq:decomposition s1}
s_1^N=\cos(a\varphi)q^N_1+\sin(a\varphi)r^N_1.
\end{equation}
By linearity of the time-ordered products, we have
\begin{equation}
\label{eq:tos1_0}
T_{l+1}\Big(L_{\text{int}}^{\otimes l}\otimes s^N_1\Big)=T_{l+1}\Big(L_{\text{int}}^{\otimes l}\otimes (\cos(a\varphi)q^N_1)\Big)+T_{l+1}\Big(L_{\text{int}}^{\otimes l}\otimes (\sin(a\varphi)r^N_1)\Big).
\end{equation} 

The two terms on the right hand side can be treated exactly in the same way, so we discuss only the second one. In this case it is not any longer true that the field $\sin(a\varphi)r^N_1$ admits non-zero functional derivatives only up to a finite order, so we cannot directly apply Proposition \ref{prop:to-int-s2}. Nevertheless we can prove a similar technical result.
\begin{prop}
\label{prop:toABC}
Consider a time-ordered product of the form
\begin{equation*}
T_{l+1}\Big(A^{\otimes l}\otimes (B\,C)\Big),\qquad A,B,C\in\mathcal{F}_{\mu c}[[\hbar]],
\end{equation*}
where $C$ is such that $\exists\,c\in\mathbb{N}$ for which $\frac{\delta^i}{\delta\varphi^i}C=0$ whenever $i>c$, while $A$ and $B$ can possibly admit non-zero functional derivatives of arbitrary order. Then the following equation hols: 
\begin{equation}
\label{eq:toABC}
\begin{split}
&T_{l+1}\big(A^{\otimes l}\otimes(B\,C)\big)= \\
&=\sum_{i=0}^{c}\hbar^i(\Delta^F)^i\sum_{\substack{i_1,\dots,i_l\ge 0 \\ i_1+\dots+i_l=i}}\frac{1}{i_1!\dots i_l!}T_{l+1}\Big(\frac{\delta^{i_1}}{\delta\varphi^{i_1}}A\otimes\dots\otimes\frac{\delta^{i_l}}{\delta\varphi^{i_l}}A\otimes B\Big)\,\frac{\delta^i}{\delta\varphi^i}C.
\end{split}
\end{equation}
\end{prop}
\begin{proof}
We start from the formula for the time-ordered products as in the proof of Proposition \ref{prop:to-int-s2}:
\begin{equation*}
T_{l+1}\big(A^{\otimes l}\otimes(B\,C)\big)=\mu\circ e^{\hbar\sum_{1\le i<j\le l}D^{ij}_F}\circ e^{\hbar\sum_{i=1}^lD^{i\,l+1}_F}\big(A^{\otimes l}\otimes(B\,C)\big). 
\end{equation*}

The second exponential acts on the fields as:
\begin{equation*}
\label{eq:multinom+}
\begin{split}
&e^{\hbar\sum_{i=1}^lD^{i\,l+1}_F}\Big(A^{\otimes l}\otimes (B\,C)\Big)= \\
&=\sum_{j=0}^\infty\frac{\hbar^j}{j!}(\Delta^F)^j\big(\frac{\delta}{\delta\varphi_1}+\dots+\frac{\delta}{\delta\varphi_l}\big)^j\otimes\frac{\delta^j}{\delta\varphi_{l+1}^j}\Big(A^{\otimes l}\otimes(B\,C)\Big).
\end{split}
\end{equation*}
Using Leibniz rule and the fact that the field $C$ admits non-zero functional derivatives only up to order $c$, we can manipulate this expression to obtain:
\begin{equation*}
\begin{split}
&e^{\hbar\sum_{i=1}^lD^{i\,l+1}_F}\Big(A^{\otimes l}\otimes (B\,C)\Big)= \\
&=\sum_{i=0}^{c}\frac{\hbar^i}{i!}(\Delta^F)^i\big(\frac{\delta}{\delta\varphi_1}+\dots+\frac{\delta}{\delta\varphi_l}\big)^i\Big(e^{\hbar\sum_{i=1}^lD^{i\,l+1}_F}\big(A^{\otimes l}\otimes B\big)\Big)\otimes\frac{\delta^i}{\delta\varphi_{l+1}^i}C.
\end{split}
\end{equation*}
Applying the remaining operators, we arrive at the aimed result.
\end{proof}
In our case we have $A=L_{\text{int}}$, $B=\sin(a\varphi)$ and $C=r^N_1$. From Remark \ref{rem:comp poly} and Proposition \ref{prop:comp_cons_curr}, we know that functional derivatives of $r^N_1$ of order greater than $2N$ are all zero. We can hence apply Proposition \ref{prop:toABC} and obtain:
\begin{equation}
\label{eq:toVr1}
\begin{split}
&T_{l+1}\Big(L_{\text{int}}^{\otimes l}\otimes (\sin(a\varphi)\cdot r^N_1)\Big)= \\
&=\sum_{i=0}^{2N}\hbar^i(\Delta^F)^i\sum_{\substack{i_1,\dots,i_l\ge 0 \\ i_1+\dots+i_l=i}}\frac{1}{i_1!\dots i_l!}T_{l+1}\Big(\frac{\delta^{i_1}}{\delta\varphi^{i_1}}L_{\text{int}}\otimes\dots\otimes\frac{\delta^{i_l}}{\delta\varphi^{i_l}}L_{\text{int}}\otimes\sin(a\varphi)\Big)\cdot \\
&\qquad\cdot\frac{\delta^i}{\delta\varphi^i}r^N_1.
\end{split}
\end{equation}
A completely analogous expression occurs for term $T_{l+1}\Big(L_{\text{int}}^{\otimes l}\otimes (\cos(a\varphi)q^N_1)\Big)$ in equation (\ref{eq:tos1_0}), with $\cos(a\varphi)$ in place of $\sin(a\varphi)$ and $q^N_1$ in place of $r^N_1$.

\subsection{Retarded components $s^N_1$}

Using formula (\ref{eq:decomposition s1}) and linearity of the retarded products, we have
\begin{equation}
\label{eq:decomposition ret s1}
R_t\Big(L_{\text{int}}^{\otimes t}\otimes s^N_1\Big)=R_t\Big(L_{\text{int}}^{\otimes t}\otimes(\cos(a\varphi)q^N_1)\Big)+R_t\Big(L_{\text{int}}^{\otimes t}\otimes(\sin(a\varphi)r^N_1)\Big).
\end{equation}
The two terms on the right hand side are completely analogous, so we consider only the second one. By formula (\ref{eq:ret field model}), we can expand the retarded product as
\begin{equation}
\label{eq:retVr1}
\begin{split}
&R_t\Big(L_{\text{int}}^{\otimes t}\otimes(\sin(a\varphi)r^N_1)\Big)= \\
&=\Big(\frac{i}{\hbar}\Big)^t\sum_{l=0}^t\binom{t}{l}(-1)^{t-l}\bar{T}_{t-l}\Big(L_{\text{int}}^{\otimes (t-l)}\Big)\star T_{l+1}\Big(L_{\text{int}}^{\otimes l}\otimes(\sin(a\varphi)r^N_1)\Big).
\end{split}
\end{equation} 

First we consider the star products $\bar{T}_{t-l}\Big(L_{\text{int}}^{\otimes (t-l)}\Big)\star T_{l+1}\Big(L_{\text{int}}^{\otimes l}\otimes(\sin(a\varphi)r^N_1)\Big)$. We substitute equation (\ref{eq:toVr1}), renaming the index $i\rightarrow j$, and then apply Proposition \ref{prop:prod_fields} to every term to obtain:
\begin{equation*}
\begin{split}
&\bar{T}_{t-l}\Big(L_{\text{int}}^{\otimes (t-l)}\Big)\star T_{l+1}\Big(L_{\text{int}}^{\otimes l}\otimes(\sin(a\varphi)r^N_1)\Big)= \\
&=\sum_{j=0}^{2N}\hbar^j\sum_{\substack{j_1,\dots,j_l\ge 0 \\ j_1+\dots+j_l=j}}\frac{1}{j_1!\dots j_l!}\sum_{i=0}^{2N-j}\frac{\hbar^i}{i!}\Big(\sum_{k=0}^\infty\frac{\hbar^k}{k!}\frac{\delta^i}{\delta\varphi^i}\frac{\delta^k}{\delta\varphi^k}\bar{T}_{t-l}\Big(L_{\text{int}}^{\otimes (t-l)}\Big)\cdot \\
&\qquad\quad\cdot(W)^k\frac{\delta^k}{\delta\varphi^k}T_{l+1}\Big(\frac{\delta^{j_1}}{\delta\varphi^{j_1}}L_{\text{int}}\otimes\dots\otimes\frac{\delta^{j_l}}{\delta\varphi^{j_l}}L_{\text{int}}\otimes\sin(a\varphi)\Big)\Big)\cdot \\
&\qquad\quad\cdot(\Delta^F)^j(W)^i\frac{\delta^i}{\delta\varphi^i}\frac{\delta^j}{\delta\varphi^j}r^N_1.
\end{split}
\end{equation*}
Then we plug this expression in equation (\ref{eq:retVr1}) and arrive at:
\begin{equation}
\label{eq:retsinr1}
\begin{split}
R_t\Big(L_{\text{int}}^{\otimes t}\otimes&(\sin(a\varphi)\cdot r^N_1)\Big)= \\
=\Big(\frac{i}{\hbar}\Big)^t&\sum_{l=0}^t\binom{t}{l}(-1)^{t-l}\sum_{j=0}^{2N}\hbar^j\sum_{\substack{j_1,\dots,j_l\ge 0 \\ j_1+\dots+j_l=j}}\frac{1}{j_1!\dots j_l!}\sum_{i=0}^{2N-j}\frac{\hbar^i}{i!}\cdot \\
&\cdot\Big(\sum_{k=0}^\infty\frac{\hbar^k}{k!}\frac{\delta^i}{\delta\varphi^i}\frac{\delta^k}{\delta\varphi^k}\bar{T}_{t-l}\Big(L_{\text{int}}^{\otimes (t-l)}\Big)(W)^k\cdot \\
&\qquad\cdot\frac{\delta^k}{\delta\varphi^k}T_{l+1}\Big(\frac{\delta^{j_1}}{\delta\varphi^{j_1}}L_{\text{int}}\otimes\dots\otimes\frac{\delta^{j_l}}{\delta\varphi^{j_l}}L_{\text{int}}\otimes\sin(a\varphi)\Big)\Big)\cdot \\
&\cdot(\Delta^F)^j(W)^i\frac{\delta^i}{\delta\varphi^i}\frac{\delta^j}{\delta\varphi^j}r^N_1.
\end{split}
\end{equation}
A completely analogous expression occurs for the term $R_t\Big(L_{\text{int}}^{\otimes t}\otimes(\cos(a\varphi)q^N_1)\Big)$ in equation (\ref{eq:decomposition ret s1}), with $\cos(a\varphi)$ in place of $\sin(a\varphi)$ and $q^N_1$ in place of $r^N_1$.

\section{Renormalization of the interacting currents}
\label{sec:renorm}

In the following sections we perform the renormalization of the retarded components $R_t\Big(L_{\text{int}}^{\otimes t}\otimes s^N_2\Big)$ and $R_t\Big(L_{\text{int}}^{\otimes t}\otimes(\sin(a\varphi)\cdot r^N_1)\Big)$ (recall that the other term in equation (\ref{eq:decomposition ret s1}) for the retarded component $s^N_1$ is completely analogous). Specifically, we show that it is possible to extend their unrenormalized expressions, as distributions defined on the subset $\check{\mathbb{M}}_2^{t+1}\subseteq\mathbb{M}_2^{t+1}$ (see formula (\ref{eq:M check})), to distributions defined on the whole space $\mathbb{M}_2^{t+1}$. More importantly, we also show that the number of counterterms is bounded by the degree of the components, i.e they are super-renormalizable.

We adopt an approach which we call piecewise renormalization. It consists of three steps: expansion of the expression to be renormalized in its elementary parts, renormalization of each one of the elementary parts separately and finally showing that reassembling the renormalized elementary parts all together gives a well-defined result.

Before starting with our program, we recall (see \cite{BR18} and \cite{BRF21} for more details) that vertex operators $V_a$, $a>0$, act on configurations $\varphi\in C^\infty(\mathbb{M}_2)$, returning a distribution (a smooth function in fact), in the following way:
\begin{equation*}
V_a(x)\colon\varphi\mapsto V_a(x)[\varphi]:=e^{ia\varphi(x)}.
\end{equation*}
The critical property of vertex operators is that functional derivatives of vertex operators have the form:
\begin{equation}
\label{eq:derVert}
\frac{\delta^k}{\delta\varphi(y_1)\dots\delta\varphi(y_k)}V_a(x)=(ia)^k\delta(y_1-x)\dots\delta(y_k-x)V_a(x),
\end{equation}
thus they are essentially again vertex operators, modulo constant coefficients. We use vertex operators to write the interaction Lagrangian of the sine-Gordon model as $L_{\text{int}}=\cos(a\varphi)=\frac{1}{2}(V_a+V_{-a})$ and to write $\sin(a\varphi)=\frac{1}{2i}(V_a-V_{-a})$.

We also introduce some notation. By slight abuse, we denote with $x:=(\tau,\xi)$ the set of light-cone coordinates on $\mathbb{M}_2$, and consequently with $(x_1,\dots,x_n):=(\tau_1,\xi_1,\dots,\tau_n,\xi_n)$ the set of light-cone coordinates on $\mathbb{M}_2^n$.
Substituting the vertex operators, using formula (\ref{eq:derVert}) and omitting the numerical coefficients, we can write the generic term of equation (\ref{eq:rets2}) in the form:
\begin{equation}
\label{eq:genterms2}
\begin{split}
&\frac{\delta^i}{\delta\varphi^i}\frac{\delta^k}{\delta\varphi^k}\bar{T}_{t-l}\Big(V_{a_{l+1}}(x_{l+1})\otimes\dots\otimes V_{a_t}(x_t)\Big)\big(W(x_{\{l+1\le\cdot\le t\}}-x_{\{1\le\cdot\le l\}})\big)^k\cdot \\
&\cdot\frac{\delta^k}{\delta\varphi^k} T_l\Big(V_{a_1}(x_1)\otimes\dots\otimes V_{a_l}(x_l)\Big)\big(\partial^\cdot_{\xi_{t+1}}W(x_{\{l+1\le\cdot\le t\}}-x_{t+1})\big)^i\cdot \\
&\cdot\big(\partial^\cdot_{\xi_{t+1}}\Delta^F(x_{\{1\le\cdot\le l\}}-x_{t+1})\big)^j\frac{\delta^i}{\delta\varphi^i}\frac{\delta^j}{\delta\varphi^j}s^N_2(x_{t+1}),
\end{split}
\end{equation}
where, moreover, we have that:
\begin{itemize}
\item $a_1,\dots,a_t\in\{+a,-a\}$, $a\in\mathbb{R_+}$;
\item $\big(W(x_{\{l+1\le\cdot\le t\}}-x_{\{1\le\cdot\le l\}})\big)^k$ denotes products of $k$ Wightman two-point functions and for each one of them the first argument belongs to the set $x_{\{l+1\le\cdot\le t\}}:=\{x_{l+1},\dots,x_t\}$, while the second argument belongs to the set $x_{\{1\le\cdot\le l\}}:=\{x_1,\dots,x_l\}$, in all the possible combinations;
\item $\big(\partial^\cdot_{\xi_{t+1}}W(x_{\{l+1\le\cdot\le t\}}-x_{t+1})\big)^i$ denotes products of $i$ Wightman two-point functions where for each one of them the first argument belongs to the set $x_{\{l+1\le\cdot\le t\}}$, the second argument is $x_{t+1}$ and moreover we take a derivative of unspecified order $\partial^\cdot_{\xi_{t+1}}$ with respect to the coordinate $\xi_{t+1}$ of $x_{t+1}=(\tau_{t+1},\xi_{t+1})$ (the order of the derivative depends on the results of the functional derivatives $\frac{\delta^i}{\delta\varphi^i}$ of $s_2^N(x_{t+1})$);
\item $\big(\partial^\cdot_{\xi_{t+1}}\Delta^F(x_{\{1\le\cdot\le l\}}-x_{t+1})\big)^j$ denotes products of $j$ Feynman propagators where for each one of them the first argument belongs to the set $x_{\{1\le\cdot\le l\}}$, the second argument is $x_{t+1}$ and moreover we take a derivative of unspecified order $\partial^\cdot_{\xi_{t+1}}$ with respect to the coordinate $\xi_{t+1}$ of $x_{t+1}=(\tau_{t+1},\xi_{t+1})$ (the order of the derivative depends on the results of the functional derivatives $\frac{\delta^j}{\delta\varphi^j}$ of $s_2^N(x_{t+1})$).
\end{itemize}
For the generic term of equation (\ref{eq:retsinr1}) we have the only difference that also the time-ordered products depend on $x_{t+1}$. In this case we have:
\begin{equation}
\label{eq:gentermr1}
\begin{split}
&\frac{\delta^i}{\delta\varphi^i}\frac{\delta^k}{\delta\varphi^k}\bar{T}_{t-l}\Big(V_{a_{l+2}}(x_{l+2})\otimes\dots\otimes V_{a_{t+1}}(x_{t+1})\Big)\cdot \\
&\cdot\big(W(x_{\{l+2\le\cdot\le t+1\}}-x_{\{1\le\cdot\le l+1\}})\big)^k\big(\partial^\cdot_{\xi_{l+1}}W(x_{\{l+2\le\cdot\le t+1\}}-x_{l+1})\big)^i\cdot \\
&\cdot\frac{\delta^k}{\delta\varphi^k} T_{l+1}\Big(V_{a_1}(x_1)\otimes\dots\otimes V_{a_l}(x_l)\otimes V_{a_{l+1}}(x_{l+1})\Big)\cdot \\
&\cdot\big(\partial^\cdot_{\xi_{l+1}}\Delta^F(x_{\{1\le\cdot\le l\}}-x_{l+1})\big)^j\frac{\delta^i}{\delta\varphi^i}\frac{\delta^j}{\delta\varphi^j}r^N_1(x_{l+1}),
\end{split}
\end{equation}
with same notations as above.

We are now ready to discuss the renormalization of the generic terms (\ref{eq:genterms2}) and (\ref{eq:gentermr1}). In Section \ref{sec:renormalization to derF} we consider time-ordered products of vertex operators together with derivatives of Feynman propagators and perform what we call their piecewise renormalization. In Section \ref{sec:piecewise ren ac} we consider anti-chronological products of vertex operators and we perform their piecewise renormalization.

As for the derivatives of the components of the currents, they are smooth functions hence there is no need for renormalization. The products of Wightman two-point functions and their derivatives, instead, are always well-defined distributions on $\mathbb{M}_2^{t+1}$ according to H\"ormander's sufficient criterion.

Finally, in Section \ref{sec:big_prod} we reassemble the pieces all together and show, by a careful study of the wavefront sets of all the elements involved, that the result is well-defined.

\subsection{Piecewise renormalization of time-ordered products \\ and derivatives of Feynman propagators}
\label{sec:renormalization to derF}
According to our plan, we first expand the expressions to be renormalized in their most elementary parts. In \cite{BR18} it is shown that the unrenormalized time-ordered products of vertex operators can be written in exponential form as:
\begin{equation}
\label{eq:to_exp}
T_l\Big(V_{a_1}(x_1)\otimes\dots\otimes V_{a_l}(x_l)\Big)=e^{i\big(a_1\varphi(x_1)+\dots+a_l\varphi(x_l)\big)}\prod_{1\le i<j\le l}e^{-a_ia_j\hbar\Delta^F(x_i-x_j)}.
\end{equation}

Omitting the exponentials of configurations, which do not need renormalization, we can expand the exponentials of Feynman propagators as a formal power series in $\hbar$. The coefficient of the power $\hbar^p$ is given by:
\begin{equation}
\label{eq:exp_to_ordp}
\begin{split}
\sum_{\substack{\{p_{i,j}\ge 0,\,1\le i< j\le l \\ \text{s.t.}\,\sum_{i,j}p_{i,j}=p\}}}&\frac{(-1)^p(a_1a_2)^{p_{1,2}}\cdots(a_{l-1}a_l)^{p_{l-1,l}}}{p_{1,2}!\cdots p_{l-1,l}!}\cdot \\
&\cdot(\Delta^F)^{p_{1,2}}(x_1-x_2)\cdots(\Delta^F)^{p_{l-1,l}}(x_{l-1}-x_l),
\end{split}
\end{equation} 

\subsubsection{Discussion for components $s^N_2$}
\label{sec:piecewise ren s2}

We now concentrate specifically on the time-ordered products of vertex operators and derivative of Feynman propagators appearing in formula (\ref{eq:genterms2}). We write the products of derivatives of Feynman propagators as:
\begin{equation}
\label{eq:products of der of Feynman}
\begin{split}
&\partial_{\xi_{t+1}}^{i_{1,1}}\Delta^F(x_1-x_{t+1})\cdots\partial_{\xi_{t+1}}^{i_{1,n_1}}\Delta^F(x_1-x_{t+1})\cdot \\
\cdot&\partial_{\xi_{t+1}}^{i_{2,1}}\Delta^F(x_2-x_{t+1})\cdots\partial_{\xi_{t+1}}^{i_{2,n_2}}\Delta^F(x_2-x_{t+1})\cdot \\
&\quad\vdots \\
\cdot&\partial_{\xi_{t+1}}^{i_{l,1}}\Delta^F(x_l-x_{t+1})\cdots\partial_{\xi_{t+1}}^{i_{l,n_l}}\Delta^F(x_l-x_{t+1}),
\end{split}
\end{equation}
where $i_{r\,s}\ge 1$, if $n_r\ge 1$, and otherwise, if $n_r=0$, then there is no product of Feynman propagators with argument $(x_r-x_{t+1})$. Using formulas (\ref{eq:products of der of Feynman}) and (\ref{eq:exp_to_ordp}) in expression (\ref{eq:genterms2}), we obtain that the coefficient of $\hbar^p$, modulo multiplicative constants, is:
\begin{equation}
\label{eq:ren_to_s2}
\begin{split}
\sum_{\substack{\{p_{i,j}\ge 0,\,1\le i< j\le l \\ \text{s.t.}\,\sum_{i,j}p_{i,j}=p\}}}&\frac{(-1)^p(a_1a_2)^{p_{1,2}}\cdots(a_{l-1}a_l)^{p_{l-1,l}}}{p_{1,2}!\cdots p_{l-1,l}!}\cdot \\
&\cdot(\Delta^F)^{p_{1,2}}(x_1-x_2)\cdots(\Delta^F)^{p_{l-1,l}}(x_{l-1}-x_l)\cdot \\
&\cdot\partial_{\xi_{t+1}}^{i_{1,1}}\Delta^F(x_1-x_{t+1})\cdots\partial_{\xi_{t+1}}^{i_{1,n_1}}\Delta^F(x_1-x_{t+1})\cdot \\
&\quad\vdots \\
&\cdot\partial_{\xi_{t+1}}^{i_{l,1}}\Delta^F(x_l-x_{t+1})\cdots\partial_{\xi_{t+1}}^{i_{l,n_l}}\Delta^F(x_l-x_{t+1}).
\end{split}
\end{equation}
We consider each one of the elementary parts separately, as distributions defined on $ \mathbb{M}_2\setminus\Set{0}$, and denote them by:
\begin{equation}
\label{eq:pieces}
\begin{rcases}
D_{1,2}:=(\Delta^F)^{p_{1,2}}, \\
\quad\vdots \\
D_{l-1,l}:=(\Delta^F)^{p_{l-1,l}}, \\
D_{1,t+1}:=(\partial_{\xi_{t+1}}^{i_{1,1}}\Delta^F)\cdots(\partial_{\xi_{t+1}}^{i_{1,n_1}}\Delta^F), \\
\quad\vdots \\
D_{l,t+1}:=(\partial_{\xi_{t+1}}^{i_{l,1}}\Delta^F)\cdots(\partial_{\xi_{t+1}}^{i_{l,n_l}}\Delta^F),
\end{rcases}
\in\mathscr{D}'(\mathbb{M}_2\setminus\{0\}).
\end{equation}

In order to proceed with the second step of the piecewise renormalization process of formula (\ref{eq:ren_to_s2}), we take into account the Steinmann scaling degree of the Feynman propagator. On the $2$-dimensional Minkowski space $\mathbb{M}_2$, the Feynman propagator $\Delta^F$ scales homogeneously with scaling degree $\text{sd}(\Delta^F)=0$, and every derivative potentially increases by one the scaling degree (see \cite{D}, \cite{Stein}).

\begin{lemma}
\label{lemma:sd bound s2}
For the products of derivatives of Feynman propagators appearing in formula (\ref{eq:pieces}) the following estimate on the scaling degrees holds:
\begin{equation}
\label{eq:estimate sd}
\begin{rcases}
\text{sd}\Big((\partial_{\xi_{t+1}}^{i_{1,1}}\Delta^F)\cdots(\partial_{\xi_{t+1}}^{i_{1,n_1}}\Delta^F)\Big)\le\sum_{s=1}^{n_1}i_{1,s}, \\
\quad\vdots \\
\text{sd}\Big((\partial_{\xi_{t+1}}^{i_{l,1}}\Delta^F)\cdots(\partial_{\xi_{t+1}}^{i_{l,n_l}}\Delta^F)\Big)\le\sum_{s=1}^{n_l}i_{l,s},
\end{rcases}
\le\text{deg}(s^N_2)=2(N+1),
\end{equation}
for every $N\in\mathbb{N}$.
\end{lemma}
\begin{proof}
The result follows from Remark \ref{rem:comp poly} and Proposition \ref{prop:deg_cons_curr} on the structure of the components $s^N_2$ and from the general fact that, given two distributions $u,v\in\mathscr{D}'$ whose distributional product is well-defined, it holds
\begin{equation}
\label{eq:sd prod}
\text{sd}(uv)\le\text{sd}(u)+\text{sd}(v).
\end{equation}
\end{proof}
\begin{rem}
Formula (\ref{eq:sd prod}) also implies that any power of the Feynman propagator has scaling degree equal to $0$ on $\mathbb{M}_2\setminus\{0\}$.
\end{rem}

Knowing an estimate on the scaling degree of each one of the elements in formula (\ref{eq:pieces}), we can apply well-known results (\cite{BF}, \cite{BFK}, \cite{D}) to extend them to distributions defined on the whole $\mathbb{M}_2$ in such a way to also preserve the scaling degree. We denote these extensions by:
\begin{equation}
\label{eq:exted_pieces}
\begin{rcases}
[D_{1,2}]:=[(\Delta^F)^{p_{1,2}}], \\
\quad\vdots \\
[D_{l-1,l}]:=[(\Delta^F)^{p_{l-1,l}}], \\
[D_{1,t+1}]:=[(\partial_{\xi_{t+1}}^{i_{1,1}}\Delta^F)\cdots(\partial_{\xi_{t+1}}^{i_{1,n_1}}\Delta^F)], \\
\quad\vdots \\
[D_{l,t+1}]:=[(\partial_{\xi_{t+1}}^{i_{l,1}}\Delta^F)\cdots(\partial_{\xi_{t+1}}^{i_{l,n_l}}\Delta^F)],
\end{rcases}
\in\mathscr{D}'(\mathbb{M}_2).
\end{equation}
\begin{rem}
\label{rem:uniq of ext}
We note that for powers of Feynman propagators, since their scaling degree is equal to $0$, the extensions are direct and unique. For products of derivatives of Feynman propagators, when the scaling degree is $\text{sd}\ge 2$, the extension is unique up to adding a finite number of derivatives of the Dirac delta, namely derivatives up to order $\text{sd}-2$ (see also beginning of Section \ref{sec:unrenorm}).
\end{rem}

\begin{rem}
Considering the wavefront set of the Feynman propagator and the fact, pointed out in Remark \ref{rem:uniq of ext}, that the extensions are realized by possibly adding Dirac deltas, we have that the wavefront set of each element of (\ref{eq:exted_pieces}) is contained in the set
\begin{equation}
\label{eq:Gamma0}
\begin{split}
\Gamma_0:=&\Set{(w,k)\in T^\ast\mathbb{M}_2 |\, |w|^2=0,\,w\neq 0,\,k=\frac{\eta_\flat(w)}{\lambda},\,\lambda>0}\,\cup \\
&\cup\,\Set{(0,k)\in T^\ast\mathbb{M}_2 | k\neq 0},
\end{split}
\end{equation} 
where $|w|^2=\eta(w,w)$, and $\eta_\flat\colon T\mathbb{M}_2\rightarrow T^\ast\mathbb{M}_2$ is the isomorphism induced by the Minkowski metric.
\end{rem}

The piecewise renormalization process has to maintain the translation invariance of the unrenormalized expressions. Starting from the extended elementary parts (\ref{eq:exted_pieces}), we obtain translation-invariant distributions defined on $\mathbb{M}_2^2$ performing the pull-back of every element via appropriate maps.
\begin{lemma}
\label{lemma:well pb}
Consider the maps
\begin{equation*}
\begin{split}
s_{i,j}\colon\quad\mathbb{M}_2^2\quad&\rightarrow\quad\mathbb{M}_2 \\
(x_i,x_j)&\mapsto w_{i,j}=x_i-x_j,
\end{split}
\end{equation*}
where $i,j\in\Set{1,\dots,l}$ and $i<j$, or $i=1,\dots,l$ and $j=t+1$. Then the pull-back
\begin{equation}
\label{eq:pullb_exted_pieces}
\begin{rcases}
s_{1,2}^\ast\big([D_{1,2}]\big)=[(\Delta^F)^{p_{1,2}}](x_1-x_2), \\
\quad\vdots \\
s_{l-1,l}^\ast\big([D_{l-1,l}]\big)=[(\Delta^F)^{p_{l-1,l}}](x_{l-1}-x_l), \\
s_{1,t+1}^\ast\big([D_{1,t+1}]\big)=[(\partial_{\xi_{t+1}}^{i_{1,1}}\Delta^F)\cdots(\partial_{\xi_{t+1}}^{i_{1,n_1}}\Delta^F)](x_1-x_{t+1}), \\
\quad\vdots \\
s_{l,t+1}^\ast\big([D_{l,t+1}]\big)=[(\partial_{\xi_{t+1}}^{i_{l,1}}\Delta^F)\cdots(\partial_{\xi_{t+1}}^{i_{l,n_l}}\Delta^F)](x_l-x_{t+1}),
\end{rcases}
\in\mathscr{D}'(\mathbb{M}_2^2),
\end{equation}
are well-defined translation-invariant distributions.
\end{lemma}
\begin{proof}
Pull-back along the maps $s_{i,j}$ is well-defined in general for any distribution, hence a fortiori for our distributions (\ref{eq:exted_pieces}). Indeed the transpose of the tangent maps $(s_{i,j}')^t$ have the form:
\begin{equation*}
\begin{split}
(s_{i,j}')^t\colon\quad\qquad T^\ast\mathbb{M}_2\qquad\,&\rightarrow\qquad T^\ast\mathbb{M}_2^2 \\
(w_{i,j}=x_i-x_j,k)&\mapsto\,(x_i,k;x_j,-k).
\end{split}
\end{equation*}
Hence the condition ensuring the well-posedness of the pull-back (see \cite{Horm})
\begin{equation*}
(s_{i,j}')^t\big(\Gamma_0\big)\,\cap\,\Set{(x_i,0;x_j,0)\subset T^\ast\mathbb{M}_2^2}=\emptyset
\end{equation*}
is always satisfied, for any $i,j$ as above.
\end{proof}

\begin{rem}
From the properties of the wavefront set under the operation of pull-back and formula (\ref{eq:Gamma0}), we obtain that the wavefront set of the distributions (\ref{eq:pullb_exted_pieces}) are contained respectively in the sets
\begin{equation}
\label{eq:Gamma}
\begin{split}
\Gamma_{i,j}=&(s_{i,j}')^t\big(\Gamma_0\big)=\Big\{(x_i,k;x_j,-k)\in T^\ast\mathbb{M}_2^2 |\,|x_i-x_j|^2=0,\,x_i\neq x_j, \\
&k=\frac{\eta_\flat(x_i-x_j)}{\lambda},\,\lambda>0\Big\}\,\cup\,\Set{(x,k;x,-k)\in T^\ast\mathbb{M}_2^2 | k\neq 0},
\end{split}
\end{equation}
with $i,j\in\Set{1,\dots,l}$ and $i<j$, or $i=1,\dots,l$ and $j=t+1$.
\end{rem}

We have thus completed the piecewise renormalization of the elementary parts. Reassembling them together, we arrive at the following piecewise renormalized expression for the coefficient of $\hbar^p$:
\begin{equation}
\label{eq:ren_gen_term_to_s2}
\begin{split}
\sum_{\substack{\{p_{i,j}\ge 0,\,1\le i< j\le l \\ \text{s.t.}\,\sum_{i,j}p_{i,j}=p\}}}&\frac{(-1)^p(a_1a_2)^{p_{1,2}}\cdots(a_{l-1}a_l)^{p_{l-1,l}}}{p_{1,2}!\cdots p_{l-1,l}!}\cdot \\
&\cdot[(\Delta^F)^{p_{1,2}}](x_1-x_2)\cdots[(\Delta^F)^{p_{l-1,l}}](x_{l-1}-x_l)\cdot \\
&\cdot[(\partial_{\xi_{t+1}}^{i_{1,1}}\Delta^F)\cdots(\partial_{\xi_{t+1}}^{i_{1,n_1}}\Delta^F)](x_1-x_{t+1})\cdot \\
&\quad\vdots \\
&\cdot[(\partial_{\xi_{t+1}}^{i_{l,1}}\Delta^F)\cdots(\partial_{\xi_{t+1}}^{i_{l,n_l}}\Delta^F)](x_l-x_{t+1}).
\end{split}
\end{equation}

We can repeat the piecewise renormalization in the same way for the coefficients of every power of $\hbar^p$. Summing together all the orders, we finally obtain the piecewise renormalized expression of the time-ordered products of vertex operators and derivatives of Feynman propagators appearing in formula (\ref{eq:genterms2})
\begin{equation}
\label{eq:rened_pairing_to_s2}
\begin{split}
&\big[T_l\Big(V_{a_1}(x_1)\otimes\dots\otimes V_{a_l}(x_l)\Big)\big]\cdot[(\partial_{\xi_{t+1}}^{i_{1\,1}}\Delta^F)\cdots(\partial_{\xi_{t+1}}^{i_{1\,n_1}}\Delta^F)](x_1-x_{t+1})\cdots \\
&\cdots[(\partial_{\xi_{t+1}}^{i_{l\,1}}\Delta^F)\cdots(\partial_{\xi_{t+1}}^{i_{l\,n_l}}\Delta^F)](x_l-x_{t+1}).
\end{split}
\end{equation}
We show that the distributional products in this formula are actually well-defined according to H\"ormander's sufficent criterion in Theorem \ref{theo:main} below.
\begin{rem}
\label{rem:super s2}
We stress that, from Remark \ref{rem:uniq of ext}, the extensions regarding the Feynman propagators and consequently the time-ordered products of
vertex operators are unique and direct. Instead, the extensions regarding the products of derivatives of Feynman propagators are not unique. Nevertheless, the number of derivatives of Dirac deltas that can be added in the extension process is finite and bounded by $\text{deg}(s^N_2)-2=2N$. 
\end{rem}

\subsubsection{Discussion for components $s^N_1$} 

The piecewise renormalization of the time-ordered products of vertex operators and derivatives of Feynman propagators in the case of the retarded components $s^N_1$ proceeds in a completely analogous way. We only point out some slight differences. 

The first one is that the time-ordered products appearing in formula (\ref{eq:gentermr1}) depend also on the same argument $x_{l+1}$ as the term $r^N_1$. This translates in the fact that the coefficient of the power $\hbar^p$ takes the form
\begin{equation}
\label{eq:gen_term_to_r1}
\begin{split}
\sum_{\substack{\{p_{i,j}\ge 0,\,1\le i< j\le l+1 \\ \text{s.t.}\,\sum_{i,j}p_{i,j}=p\}}}&\frac{(-1)^p(a_1a_2)^{p_{1,2}}\cdots(a_la_{l+1})^{p_{l,l+1}}}{p_{1,2}!\cdots p_{l,l+1}!}\cdot \\
&\cdot(\Delta^F)^{p_{1,2}}(x_1-x_2)\cdots(\Delta^F)^{p_{l,l+1}}(x_l-x_{l+1})\cdot \\
&\cdot\partial_{\xi_{l+1}}^{i_{1,1}}\Delta^F(x_1-x_{l+1})\cdots\partial_{\xi_{l+1}}^{i_{1,n_1}}\Delta^F(x_1-x_{l+1})\cdot \\
&\quad\vdots \\
&\cdot\partial_{\xi_{l+1}}^{i_{l,1}}\Delta^F(x_l-x_{l+1})\cdots\partial_{\xi_{l+1}}^{i_{l,n_l}}\Delta^F(x_l-x_{l+1}).
\end{split}
\end{equation}

The second difference is that, from Proposition \ref{prop:deg_cons_curr}, $\text{deg}(r^N_1)=2N$. So the same argument as in Lemma \ref{lemma:sd bound s2} tells us that the scaling degrees of the products of derivatives of Feynman propagators in (\ref{eq:gen_term_to_r1}) are bounded by $2N$. 

We can repeat the same passages as in the previous Section, obtaining a piecewise renormalized version of the coefficient of $\hbar^p$ for every order $p$:
\begin{equation}
\label{eq:piecewiese ren gen_term_to_r1}
\begin{split}
\sum_{\substack{\{p_{i,j}\ge 0,\,1\le i< j\le l+1 \\ \text{s.t.}\,\sum_{i,j}p_{i,j}=p\}}}&\frac{(-1)^p(a_1a_2)^{p_{1,2}}\cdots(a_la_{l+1})^{p_{l,l+1}}}{p_{1,2}!\cdots p_{l,l+1}!}\cdot \\
&\cdot[(\Delta^F)^{p_{1,2}}](x_1-x_2)\cdots[(\Delta^F)^{p_{l,l+1}}](x_l-x_{l+1})\cdot \\
&\cdot[(\partial_{\xi_{l+1}}^{i_{1,1}}\Delta^F)\cdots(\partial_{\xi_{l+1}}^{i_{1,n_1}}\Delta^F)](x_1-x_{l+1})\cdot \\
&\quad\vdots \\
&\cdot[(\partial_{\xi_{l+1}}^{i_{l,1}}\Delta^F)\cdots(\partial_{\xi_{l+1}}^{i_{l,n_l}}\Delta^F)](x_l-x_{l+1}).
\end{split}
\end{equation}
Summing together all the orders, we finally get the piecewise renormalized expression of the time-ordered products of vertex operators and derivatives of Feynman propagators appearing in formula (\ref{eq:gentermr1}):
\begin{equation}
\label{eq:rened_pairing_to_r1}
\begin{split}
&\big[T_{l+1}\Big(V_{a_1}(x_1)\otimes\dots\otimes V_{a_l}(x_l)\otimes V_{a_{l+1}}(x_{l+1})\Big)\big]\cdot \\
&[(\partial_{\xi_{l+1}}^{i_{1,1}}\Delta^F)\cdots(\partial_{\xi_{l+1}}^{i_{1,n_1}}\Delta^F)](x_1-x_{l+1})\cdots[(\partial_{\xi_{l+1}}^{i_{l,1}}\Delta^F)\cdots(\partial_{\xi_{l+1}}^{i_{l,n_l}}\Delta^F)](x_l-x_{l+1}).
\end{split}
\end{equation}
We show that the distributional products in this formula are actually well-defined according to H\"ormander's sufficent criterion in Theorem \ref{theo:main} below.
\begin{rem}
As previously pointed out in Remark \ref{rem:super s2} for the retarded components $s^N_2$, also in this case the ambiguity in the renormalization process, given by number of derivatives of Dirac deltas that can be added, is finite and bounded by $\text{deg}(s^N_1)-2=2(N-1)$.
\end{rem}

\subsection{Piecewise renormalization of anti-chronological \\ products of vertex operators}
\label{sec:piecewise ren ac}

We now consider the anti-chronological products appearing in formulas (\ref{eq:genterms2}) and (\ref{eq:gentermr1}). In both cases we have anti-chronological products of vertex operators with $t-l$ arguments. We indicate them by:
\begin{equation*}
\label{eq:unrenorm ac}
\bar{T}_{t-l}\Big(V_{a_{l+1}}(x_{l+1})\otimes\dots\otimes V_{a_t}(x_t)\Big).
\end{equation*}
The anti-chronological products of vertex operators can be written in the following exponential form:
\begin{equation*}
\label{eq:anticr_exp}
\begin{split}
&\bar{T}_{t-l}\Big(V_{a_{l+1}}(x_{l+1})\otimes\dots\otimes V_{a_t}(x_t)\Big)= \\
&=e^{i\big(a_{l+1}\varphi(x_{l+1})+\dots+a_t\varphi(x_t)\big)}\prod_{l+1\le i<j\le t}e^{-a_ia_j\hbar\Delta^{AF}(x_i-x_j)}, 
\end{split}
\end{equation*}
where $\Delta^{AF}$ is the anti Feynman propagator, defined as the complex conjugate of the Feynman propagator $\Delta^{AF}=\overline{\Delta^F}$. For its scaling degree, it holds:
\begin{equation*}
\label{eq:sd antiF}
\text{sd}(\Delta^{AF})=\text{sd}(\overline{\Delta^F})=\text{sd}(\Delta^F)=0.
\end{equation*}

We can expand the product of exponentials of anti Feynman propagators and collect the coefficient of $\hbar^q$:
\begin{equation*}
\label{eq:anticr_ordq}
\begin{split}
\sum_{\substack{\{q_{i,j}\ge 0,\,l+1\le i<j\le t \\ \text{s.t.}\,\sum_{i,j}q_{i,j}=q}\}}&\frac{(-1)^q(a_{l+1}a_{l+2})^{q_{l+1,l+2}}\cdots(a_{t-1}a_{t})^{q_{t-1,t}}}{q_{l+1,l+2}!\cdots q_{t-1,t}!}\cdot \\
&\cdot(\Delta^{AF})^{q_{l+1,l+2}}(x_{l+1}-x_{l+2})\cdots(\Delta^{AF})^{q_{t-1,t}}(x_{t-1}-x_{t}).
\end{split}
\end{equation*}
Each one of the elementary parts has scaling degree equal to $0$. So they admit direct and unique extensions (cfr. Remark \ref{rem:uniq of ext}). Repeating once more the same passages as in Section \ref{sec:piecewise ren s2}, we arrive at the piecewise renormalized expression of the coefficient of $\hbar^q$:
\begin{equation}
\label{eq:ren_anticr_ordq}
\begin{split}
\sum_{\substack{\{q_{i,j}\ge 0,\,l+1\le i<j\le t \\ \text{s.t.}\,\sum_{i,j}q_{i,j}=q}\}}&\frac{(-1)^q(a_{l+1}a_{l+2})^{q_{l+1,l+2}}\cdots(a_{t-1}a_{t})^{q_{t-1,t}}}{q_{l+1,l+2}!\cdots q_{t-1,t}!}\cdot \\
&\cdot[(\Delta^{AF})^{q_{l+1,l+2}}](x_{l+1}-x_{l+2})\cdots[(\Delta^{AF})^{q_{t-1,t}}](x_{t-1}-x_{t}).
\end{split}
\end{equation}
Summing together all the orders, we finally get the piecewise renormalized expression of the anti-chronological products of vertex operators appearing in formulas (\ref{eq:genterms2}) and (\ref{eq:gentermr1}), which we denote by:
\begin{equation}
\label{eq:ren_anticr_prod}
\big[\bar{T}_{t-l}\Big(V_{a_{l+1}}(x_{l+1})\otimes\dots\otimes V_{a_t}(x_t)\Big)\big].
\end{equation}

\subsection{Well-posedness of the piecewise renormalized expressions}
\label{sec:big_prod}
We complete the renormalization of the retarded components of the higher currents showing that the piecewise renormalized expressions obtained in the previous sections are indeed well-defined. We proceed as follows.

First we prove the well-posedness of the piecewise renormalized time-ordered products of vertex operators and derivatives of Feynman propagators and of the piecewise renormalized anti-chronological products of vertex operators by explicitly computing their wavefront set. Then we estimate the wavefront set of the products of (derivatives of) Wightman two-point functions. Finally, we show that these elements satisfy H\"ormander's sufficient criterion and thus their product is well-defined. \\

Let us start with the piecewise renormalized coefficient of the power $\hbar^p$ of time-ordered products of vertex operators and derivatives of Feynman propagators for components $s^N_1$, formula (\ref{eq:piecewiese ren gen_term_to_r1}). We regard it as the result of the pull-back of the tensor product of all its elementary parts
\begin{equation}
\label{eq:tensor prod to dF}
\begin{split}
\text{TOF}^p_{l+1}:=\sum_{\substack{\{p_{i,j}\ge 0,\,1\le i< j\le l+1 \\ \text{s.t.}\,\sum_{i,j}p_{i,j}=p\}}}&\frac{(-1)^p(a_1a_2)^{p_{1,2}}\cdots(a_la_{l+1})^{p_{l,l+1}}}{p_{1,2}!\cdots p_{l,l+1}!} \\
&[(\Delta^F)^{p_{1,2}}](w_{1,2})\otimes\cdots\otimes[(\Delta^F)^{p_{l,l+1}}](w_{l,l+1})\otimes \\
\otimes&[(\partial_{\xi_{l+1}}^{i_{1,1}}\Delta^F)\cdots(\partial_{\xi_{l+1}}^{i_{1,n_1}}\Delta^F)](\tilde{w}_{1,l+1})\otimes \\
&\quad\vdots \\
\otimes&[(\partial_{\xi_{l+1}}^{i_{l,1}}\Delta^F)\cdots(\partial_{\xi_{l+1}}^{i_{l,n_l}}\Delta^F)](\tilde{w}_{l,l+1}),
\end{split}
\end{equation}
seen as a distribution in $\mathscr{D}'(\mathbb{M}_2^K)$, $K=\binom{l+1}{2}+l$, via the map 
\begin{equation}
\label{eq:s_map}
\begin{split}
s\colon\quad\qquad\mathbb{M}_2^{l+1}\qquad&\rightarrow\qquad\mathbb{M}_2^K \\
(x_1,\dots,x_{l+1})&\mapsto (w_{i,j}=x_i-x_j,\tilde{w}_{k,l+1}=x_k-x_{l+1}),
\end{split}
\end{equation}
for $1\le i<j\le l+1$ and $1\le k\le l$. 

The question of the well-posedness of the coefficient (\ref{eq:piecewiese ren gen_term_to_r1}) is now rephrased in terms of the well-posedness of the pull-back $s^\ast(\text{TOF}^p_{l+1})$. In other words, we ask whether the condition
\begin{equation*}
(s')^t\big(\text{WF}(\text{TOF}^p_{l+1})\big)\bigcap\big(\mathbb{M}_2\times\Set{0}\big)^{l+1}=\emptyset,
\end{equation*}
where $(s')^t$ is the transpose of the tangent map of $s$ and $\text{WF}(\text{TOF}^p_{l+1})$ is the wavefront set of $\text{TOF}^p_{l+1}$, is satisfied.

A graph notation, introduced in \cite{BFK} to describe the wavefront set of products of Feynman propagators in the context of Algebraic Quantum Field Theory on curved spacetimes, turns out to be the proper tool to answer this question. We recall this notation from \cite{BF}, adapting it to our spacetime $\mathbb{M}_2$:
\begin{itemize}
\item Denote by $\mathcal{G}_n$ the set of non-oriented graphs with vertexes $V=\Set{1,\dots,n}$, and by $E^G$ the set of edges of a given graph $G\in\mathcal{G}_n$. For any edge $e\in E^G$ between vertexes $i<j$, we set source $\sigma(e)=i$ and target $\tau(e)=j$;
\item A couple of maps $(\chi,\kappa)$ is an immersion of the graph $G\in\mathcal{G}_n$ into our spacetime $\mathbb{M}_2$ if:
\begin{itemize}
\item $\chi\colon V\rightarrow\mathbb{M}_2$ maps vertexes $i$ of $G$ to points $x_i\in\mathbb{M}_2$, with the condition that if the vertexes $i<j$ are connected by an edge, then $|x_i-x_j|^2=\eta(x_i-x_j,x_i-x_j)=0$;
\item $\kappa\colon E^G\rightarrow T^\ast\mathbb{M}_2$ with the condition that, if the vertexes $i<j$ are connected by the edge $e\in E^G$, then the covector $\kappa(e)=:k_e$ is 
\begin{equation}
\label{eq:k_e_F}
\begin{cases}
k_e=\lambda_{ij}\eta_\flat(x_i-x_j)\quad\text{for some}\quad\lambda_{ij}>0,\qquad\text{if}\quad x_i\neq x_j, \\
k_e\in\big(\mathbb{M}_2\setminus\Set{0}\big),\qquad\text{if}\quad x_i=x_j.
\end{cases}
\end{equation}
The covector $k_e$ is said to be outgoing for the point $x_i$ and incoming for the point $x_j$. 
\end{itemize} 
\end{itemize}
Using this notation, the wavefront set $\Lambda_{l+1}:=(s')^t\big(\text{WF}(\text{TOF}_p)\big)$ of $s^\ast(\text{TOF}^p_{l+1})$ can be described as: \begin{equation}
\label{eq:Lambda_l+1}
\begin{split}
\Lambda_{l+1}=\Biggl\{(&x_1,k_1;\dots;x_{l+1},k_{l+1})\in T^\ast\mathbb{M}_2^{l+1}\,|\,\,\exists\,G\in\mathcal{G}_{l+1}\quad\text{and} \\
&\exists\quad\text{an immersion}\quad(\chi,\kappa)\quad\text{of}\quad G\quad\text{such that} \\
&\quad k_i=\sum_{\substack{e\in E^G \\ \sigma(e)=i}}k_e-\sum_{\substack{f\in E^G \\ \tau(f)=i}}k_f\Biggr\}.
\end{split}
\end{equation}
\begin{prop}
\label{prop:well_def} 
For every order $p$, the piecewise renormalized coefficient of $\hbar^p$, formula (\ref{eq:piecewiese ren gen_term_to_r1}), is a well-defined distribution on $\mathbb{M}_2^{l+1}$. Namely, the condition
\begin{equation}
\label{eq:pull_back_cond}
\Lambda_{l+1}\cap\big(\mathbb{M}_2\times\Set{0}\big)^{l+1}=\emptyset
\end{equation}
for the well-posedness of the pull-back $s^\ast(\text{TOF}^{\,p}_{l+1})$ is satisfied. Consequently the piecewise renormalized time-ordered products of vertex operators and derivatives of Feynman propagators, formula (\ref{eq:rened_pairing_to_r1}), are also well-defined.
\end{prop}
\begin{proof}
Consider first immersed graphs in $\mathbb{M}_2$ with no loops, i.e suppose that the immersion map $\chi$ is injective. In this case we obtain the thesis from the following argument. For every immersed vertex $x_i$, the corresponding covector $k_i$ is given by a sum of covectors which are coparallel to the directions of connection of the vertex to its adjacent vertexes in the immersed graph. The directions of the connections always lie on the boundary of the light-cone. This means that, in order to have all covectors $k_i$ equal to zero, every vertex $x_i$ of the immersed graph has to be connected to its adjacent vertexes in opposite directions. But this can never be the case. In fact, each connected component of every immersed graph has a finite number of vertexes and if we consider for example, in a connected component, the vertex $\bar{x}$ with maximum time coordinate, then this vertex will be connected to its adjacent vertexes only in past-directed directions. Hence the covectors over $\bar{x}$ cannot sum up to zero.

Suppose now that the immersed graph contains loops, namely that the immersion $\chi$ maps vertexes $I=\{i_1,\dots,i_n\}\subseteq\{1,\dots,l+1\}$, $n\le l+1$, to the same point $x_{I}\in\mathbb{M}_2$. Let us denote by $E^G_I$ the set of loops, namely $E^G_I$ is the subset of edges $e\in E^G$ such that $\sigma(e)\in I$ and $\tau(e)\in I$. Then the conditions that the covectors $k_{i_1},\dots,k_{i_n}$ over the points $x_{i_1}=\dots=x_{i_n}=x_I$ are all equal to zero can be written as:
\begin{equation*}
\begin{split}
k_{i_1}&=\sum_{\substack{e\in E^G\setminus E^G_I \\ \sigma(e)=i_1}}k_e-\sum_{\substack{f\in E^G\setminus E^G_I \\ \tau(f)=i_1}}k_f+\sum_{\substack{e\in E^G_I \\ \sigma(e)=i_1}}k_e-\sum_{\substack{f\in E^G_I \\ \tau(f)=i_1}}k_f=0 \\
&\quad\vdots \\
k_{i_n}&=\sum_{\substack{e\in E^G\setminus E^G_I \\ \sigma(e)=i_n}}k_e-\sum_{\substack{f\in E^G\setminus E^G_I \\ \tau(f)=i_n}}k_f+\sum_{\substack{e\in E^G_I \\ \sigma(e)=i_n}}k_e-\sum_{\substack{f\in E^G_I \\ \tau(f)=i_n}}k_f=0.
\end{split}
\end{equation*}
From these equations we see that each one of the covectors $k_e$ associated to an edge $e\in E^G_I$ appears twice, with opposite signs. If we sum up the equations above, we are then left with the condition:
\begin{equation*}
k_I=k_{i_1}+\dots+k_{i_n}=\sum_{\substack{e\in E^G\setminus E^G_I \\ \sigma(e)\in I}}k_e-\sum_{\substack{f\in E^G\setminus E^G_I \\ \tau(f)\in I}}k_f=0.
\end{equation*}
This corresponds to the condition that we get if we look at the immersed graph $G$, without considering the loops. We are then reduced to the situation discussed above and we can apply the same argument to conclude.
\end{proof}
Completing the characterization of the renormalized time-ordered products of vertex operators and derivatives of Feynman propagators, we have the following result.
\begin{prop}
\label{prop:micro_cond}
The wavefront set $\Lambda_{l+1}$ of the renormalized time-ordered products of vertex operators and derivatives of Feynman propagators satisfies the microlocal condition
\begin{equation*}
\Lambda_{l+1}\cap\Big(\big(\mathbb{M}_2\times\overline{V}_-\big)^{l+1}\cup\big(\mathbb{M}_2\times\overline{V}_+\big)^{l+1}\Big)=\emptyset,
\end{equation*}
where $\overline{V}_-$ and $\overline{V}_+$ are the closure of the past and future light cones respectively.
\end{prop}
\begin{proof}
For each connected component of each immersed graph, we have a vertex $\bar{x}_+$ with maximum time coordinate and another vertex $\bar{x}_-$ with minimum time coordinate. This means that $\bar{x}_+$ is connected to its adjacent vertexes only by past-directed directions, and hence the covector over $\bar{x}_+$ is past-directed. Conversely the vertex $\bar{x}_-$ is connected to its adjacent vertexes only by future-directed directions, and hence the covector over $\bar{x}_-$ is future-directed.

This situation is not affected by the presence of loops at the vertexes $\bar{x}_+$ or $\bar{x}_-$. In fact, suppose that $\bar{x}_+$ is the image, via the immersion map $\chi$, of the vertexes $I:=\{i_1,\dots,i_n\}\subseteq\{1,\dots,l+1\}$, $n\le l+1$. Then, similarly as in the proof of Proposition \ref{prop:well_def}, we have that the covectors over the immersed vertexes $x_{i_1},\dots,x_{i_n}$ can be summed up to give:
\begin{equation*}
k_+=k_{i_1}+\dots+k_{i_n}=\sum_{\substack{e\in E^G\setminus E^G_I \\ \sigma(e)\in I}}k_e-\sum_{\substack{f\in E^G\setminus E^G_I \\ \tau(f)\in I}}k_f,
\end{equation*}
which is precisely the expression that we get if we look at the immersed graph $G$, without considering the loops. If we now assume that all covectors belong to $\overline{V}_+$, then also $k_+\in\overline{V}_+$. This is a contradiction, because from the argument at the beginning we know that for immersed graphs without loops the covector $k_+$ over $\bar{x}_+$ must belong to $\overline{V}_-$. If we assume, on the contrary, that all covectors belong to $\overline{V}_-$ and repeat the previous reasoning for $\bar{x}_-$, we get a contradiction since we know that for immersed graphs without loops the covector over $\bar{x}_-$ must belong to $\overline{V}_+$. 
\end{proof}

\begin{rem}
\label{rem:comment on s2}
For what concerns the piecewise renormalized time-ordered products of vertex operators and derivatives of Feynman propagators in the case of components $s^N_2$, formulas (\ref{eq:ren_gen_term_to_s2}) and (\ref{eq:rened_pairing_to_s2}), it suffices to repeat the same passages substituting the subscript $l+1$ with $t+1$.
\end{rem}
We now consider the piecewise renormalized coefficient of the power $\hbar^q$ of the anti-chronological products of vertex operators, formula (\ref{eq:ren_anticr_prod}). We regard it as the pull-back of the tensor product
\begin{equation*}
\label{eq:tensor prod ren_anticr_ordq}
\begin{split}
\text{ACV}^q_{t-l}:=&\sum_{\substack{\{q_{i,j}\ge 0,\,l+1\le i<j\le t \\ \text{s.t.}\,\sum_{i,j}q_{i,j}=q}\}}\frac{(-1)^q(a_{l+1}a_{l+2})^{q_{l+1,l+2}}\cdots(a_{t-1}a_{t})^{q_{t-1,t}}}{q_{l+1,l+2}!\cdots q_{t-1,t}!}\cdot \\
&\qquad\qquad\cdot[(\Delta^{AF})^{q_{l+1,l+2}}](w_{l+1,l+2})\otimes\cdots\otimes[(\Delta^{AF})^{q_{t-1,t}}](w_{t-1,t}),
\end{split}
\end{equation*}
as a distribution defined on $\mathbb{M}_2^{\tilde{K}}$, $\tilde{K}=\binom{t-l}{2}$, via the map
\begin{equation*}
\label{eq:s' map}
\begin{split}
\tilde{s}\colon\quad\qquad\mathbb{M}_2^{t-l}\qquad&\rightarrow\qquad\mathbb{M}_2^{\tilde{K}} \\
(x_{l+1},\dots,x_t)&\mapsto (w_{i,j}=x_i-x_j),
\end{split}
\end{equation*}
for $l+1\le i<j\le t$. The condition for the well-posedness of the pull-back $\tilde{s}\,^\ast\big(\text{ACV}^q_{t-l}\big)$ becomes then
\begin{equation*}
(\tilde{s}\,')^t\big(\text{WF}(\text{ACV}^q_{t-l})\big)\bigcap\big(\mathbb{M}_2\times\Set{0}\big)^{t-l}=\emptyset.
\end{equation*}
The set $\tilde{\Lambda}_{t-l}:=(\tilde{s}\,')^t\big(\text{WF}(\text{ACV}^q_{t-l})\big)$ can be described slightly adapting the graph notation. Recalling that the anti Feynman propagator is defined as $\Delta^{AF}=\overline{\Delta^F}$, we have the relation: 
\begin{equation*}
\label{eq:WF_anti_F}
\text{WF}(\Delta^{AF})=-\text{WF}(\Delta^{F})=\Set{(w,-k)\in T^\ast\mathbb{M}_2 | (w,k)\in\text{WF}(\Delta^F)}.
\end{equation*}
This means that in this case, in the definition of immersion $(\tilde{\chi},\tilde{\kappa})$ of a graph, the prescription is:
\begin{equation}
\label{eq:k_e_AF}
\begin{cases}
\tilde{k}_e=-\lambda_{ij}\eta_\flat(x_i-x_j)\quad\text{for some}\quad\lambda_{ij}>0,\qquad\text{if}\quad x_i\neq x_j, \\
\tilde{k}_e\in\big(\mathbb{M}_2\setminus\Set{0}\big),\qquad\text{if}\quad x_i=x_j.
\end{cases}
\end{equation}
We have then:
\begin{equation}
\label{eq:tilde L}
\begin{split}
\tilde{\Lambda}_{t-l}:=\Biggl\{(&x_{l+1},k_{l+1};\dots;x_{t},k_{t})\in T^\ast\mathbb{M}_2^{t-l}\,|\,\,\exists\,G\in\mathcal{G}_{t-l}\quad\text{and} \\
&\exists\quad\text{an immersion}\quad(\tilde{\chi},\tilde{\kappa})\quad\text{of}\quad G\quad\text{such that} \\
&\quad k_i=\sum_{\substack{e\in E^G \\ \sigma(e)=i}}\tilde{k}_e-\sum_{\substack{f\in E^G \\ \tau(f)=i}}\tilde{k}_f\Biggr\}.
\end{split}
\end{equation}
This modification does not affect the validity of the arguments in the proofs of Proposition \ref{prop:well_def} and Proposition \ref{prop:micro_cond}, whose passages can be repeated to obtain the expected results. 
\begin{prop}
For every order $q$, the piecewise renormalized coefficient of $\hbar^q$, formula (\ref{eq:ren_anticr_ordq}), is a well-defined distribution on $\mathbb{M}_2^{t-l}$. Namely, the condition
\begin{equation*}
\tilde{\Lambda}_{t-l}\cap\big(\mathbb{M}_2\times\Set{0}\big)^{t-l}=\emptyset.
\end{equation*}
for the well-posedness of the pull-back $\tilde{s}\,^\ast\big(\text{ACV}^{\,q}_{t-l}\big)$ is satisfied. Consequently the piecewise renormalized anti-chronological products of vertex operators, formula (\ref{eq:ren_anticr_prod}), are also well-defined. 
\end{prop} 
\begin{prop}
\label{prop:micro_cond_ac}
The wavefront set of the renormalized anti-chronological products of vertex operators $\tilde{\Lambda}_{t-l}$ satisfies the microlocal condition
\begin{equation*}
\tilde{\Lambda}_{t-l}\cap\Big(\big(\mathbb{M}_2\times\overline{V}_-\big)^{t-l}\cup\big(\mathbb{M}_2\times\overline{V}_+\big)^{t-l}\Big)=\emptyset.
\end{equation*}
\end{prop}

It remains to consider the products of  Wightman two-point functions and their derivatives. We recall that these products are always well-defined, hence no renormalization is needed in this case. Without loss of generality, we can consider as working example the product appearing in formula (\ref{eq:gentermr1})
\begin{equation}
\label{eq:prodW}
\big(W(x_{\{l+2\le\cdot\le t+1\}}-x_{\{1\le\cdot\le l+1\}})\big)^k\big(\partial^\cdot_{\xi_{l+1}}W(x_{\{l+2\le\cdot\le t+1\}}-x_{l+1})\big)^i.
\end{equation}
The analogous product appearing in formula (\ref{eq:genterms2}) can be treated in the same way. The only difference between formulas (\ref{eq:gentermr1}) and (\ref{eq:genterms2}) is in the way the dependence of the various elements on the coordinates $(x_1,\dots,x_{t+1})$ is distributed.

Once more, we can estimate the wavefront set of such products by means of the graph notation introduced above. The wavefront set of the Wightman two-point function $W$ is given by (\cite{D}):
\begin{equation*}
\label{eq:WFW}
\Set{(x,k)\in T^\ast\mathbb{M}_2 |\,|x|^2=0,\,|k|^2=0,\,\lambda k=\eta_\flat(x),\,\lambda\in\mathbb{R}\,\,\text{s.t.}\,\,k\in\partial\overline{V}_+\setminus\Set{0}}.
\end{equation*}
Hence we have to modify the convention (\ref{eq:k_e_F}) in the following way: for vertexes $1\le i<j\le t+1$ connected by an edge $e$, we set source $\sigma(e):=j$ and target $\tau(e):=i$, and define an immersion $(\hat{\chi},\hat{\kappa})$ by
\begin{equation*}
\label{eq:k_e_W}
\begin{cases}
\hat{k}_e=\lambda_{ji}\eta_\flat(x_j-x_i)\quad\text{with}\,\,\lambda_{ji}\in\mathbb{R}\quad\text{s.t.}\quad\hat{k}_e\in\partial\overline{V}_+\setminus\Set{0},\quad\text{if}\quad x_j\neq x_i, \\
\hat{k}_e\in \partial\overline{V}_+\setminus\Set{0},\qquad\text{if}\quad x_j=x_i.
\end{cases}
\end{equation*}
We obtain then the following description for the wavefront set $\Omega_{t+1}$ of the product of Wightman two-point functions and their derivatives (\ref{eq:prodW}):
\begin{equation*}
\begin{split}
\Omega_{t+1}=\Biggl\{(&x_{1},k_{1};\dots;x_{t+1},k_{t+1})\in T^\ast\mathbb{M}_2^{t+1}\,|\,\,\exists\,G\in\mathcal{G}_{t+1}\quad\text{and} \\
&\exists\quad\text{an immersion}\quad(\hat{\chi},\hat{\kappa})\quad\text{of}\quad G\quad\text{such that} \\
&\quad k_i=\sum_{\substack{e\in E^G \\ \sigma(e)=i}}\hat{k}_e-\sum_{\substack{f\in E^G \\ \tau(f)=i}}\hat{k}_f\Biggr\}.
\end{split}
\end{equation*}

\begin{rem}
\label{rem:estimate WF W}
Considering how the coordinates $(x_1,\dots,x_{t+1})$ are distributed in formula (\ref{eq:prodW}), we see that the vertexes $\Set{x_{l+2},\dots,x_{t+1}}$ only have outgoing edges. Conversely the vertexes $\Set{x_1,\dots,x_{l+1}}$ only have incoming edges. This means that the wavefront set $\Omega_{t+1}$ of the product (\ref{eq:prodW}) can be estimated by a more explicit expression, namely:
\begin{equation}
\label{eq:WF_prodW}
\begin{split}
\Omega_{t+1}\subseteq\tilde{\Omega}_{t+1}:=\Big\{(&x_1,k_1;\dots;x_{l+1},k_{l+1};x_{l+2},k_{l+2};\dots;x_{t+1},k_{t+1})\in T^\ast\mathbb{M}_2^{t+1} \\
&\text{s.t.}\quad k_1,\dots,k_{l+1}\in\overline{V}_-\quad\text{and}\quad k_{l+2},\dots,k_{t+1}\in\overline{V}_+\Big\}.
\end{split}
\end{equation}
\end{rem}
\begin{theorem}
\label{theo:main}
The retarded components $R_t\Big(L_{\text{int}}^{\otimes t}\otimes s^N_1\Big)$ and $R_t\Big(L_{\text{int}}^{\otimes t}\otimes s^N_2\Big)$ of the higher conserved currents of the sine-Gordon model are super-renormalizable by power counting in pAQFT.
\end{theorem}
\begin{proof}
In the previous part of this Section we have collected almost all the elements to prove our conclusive result. It is sufficient now to prove that the distributional product of renormalized anti-chronological products of vertex operators, renormalized time-ordered products of vertex operators and derivatives of Feynman propagators and Wightman two-point functions and their derivatives is well-defined on $\mathbb{M}_2^{t+1}$ according to H\"ormander's criterion.

We know from formula (\ref{eq:WF_prodW}), an explicit estimate on the wavefront set of the product of Wightman two-point functions and their derivatives. On the other hand, we can regard the product of renormalized anti-chronological products of vertex operators with renormalized time-ordered products of vertex operators and derivatives of Feynman propagators, defined respectively on $\mathbb{M}_2^{t-l}$ and $\mathbb{M}_2^{l+1}$, as a tensor product of distributions. We denote it by:
\begin{equation*}
\label{eq:prod_ac_to_Fp}
\begin{split}
\text{ACV}_{t-l}\otimes\text{TOF}_{l+1}:=\Big(&\big[\bar{T}_{t-l}\big(V_{a_{l+2}}(x_{l+2})\otimes\dots\otimes V_{a_{t+1}}(x_{t+1})\big)\big]\Big)\otimes \\
\otimes\Big(&\big[T_{l+1}\big(V_{a_1}(x_1)\otimes\dots\otimes V_{a_l}(x_l)\otimes V_{a_{l+1}}(x_{l+1})\big)\big]\cdot \\
\cdot&[(\partial_{\xi_{l+1}}^{i_{1,1}}\Delta^F)\cdots(\partial_{\xi_{l+1}}^{i_{1,n_1}}\Delta^F)](x_1-x_{l+1})\cdots \\
&\cdots[(\partial_{\xi_{l+1}}^{i_{l,1}}\Delta^F)\cdots(\partial_{\xi_{l+1}}^{i_{l,n_l}}\Delta^F)](x_l-x_{l+1})\Big)\in\mathscr{D}'\big(\mathbb{M}_2^{t+1}\big).
\end{split}
\end{equation*}
From the properties of the tensor product of distributions (\cite{Horm}), we have that the wavefront set of $\text{ACV}_{t-l}\otimes\text{TOF}_{l+1}$ is contained in the set:
\begin{equation*}
\label{eq:WF_prod_ac_to_Fp}
\begin{split}
\Lambda_{l+1,t-l}:=\bigg(&\Lambda_{l+1}\times\tilde{\Lambda}_{t-l}\bigg)\cup\bigg(\Lambda_{l+1}\times\big(\mathbb{M}_2\times\{0\}\big)^{t-l}\bigg)\,\cup \\
&\cup\bigg(\big(\mathbb{M}_2\times\{0\}\big)^{l+1}\times\tilde{\Lambda}_{t-l}\bigg)\subseteq T^\ast\mathbb{M}_2^{t+1},
\end{split}
\end{equation*}
where $\Lambda_{l+1}$ and $\tilde{\Lambda}_{t-l}$ are defined by formula (\ref{eq:Lambda_l+1}) and formula (\ref{eq:tilde L}) respectively. From Proposition \ref{prop:micro_cond} and Proposition \ref{prop:micro_cond_ac}, we have that
\begin{equation}
\label{eq:microlocal prod}
\begin{split}
\Lambda_{l+1,t-l}\cap\bigg(\Big(\big(&\mathbb{M}_2\times\bar{V}_-\big)^{l+1}\cup\big(\mathbb{M}_2\times\bar{V}_+\big)^{l+1}\Big)\times \\
&\times\Big(\big(\mathbb{M}_2\times\bar{V}_-\big)^{t-l}\cup\big(\mathbb{M}_2\times\bar{V}_+\big)^{t-l}\Big)\bigg)=\emptyset.
\end{split}
\end{equation}
If we now consider the set
\begin{equation*}
\begin{split}
\Lambda_{l+1,t-l}+\Omega_{t+1}:=\big\{&(x_1,r_1+s_1;\dots;x_{t+1},r_{t+1}+s_{t+1})\in T^\ast\mathbb{M}_2^{t+1}\,| \\ &(x_1,r_1;\dots;x_{t+1},r_{t+1})\in\Lambda_{l+1,t-l}, \\
&\text{and}\,(x_1,s_1;\dots;x_{t+1},s_{t+1})\in\Omega_{t+1}\big\}
\end{split}
\end{equation*}
and compare formula (\ref{eq:WF_prodW}) and formula (\ref{eq:microlocal prod}), we see immediately that
\begin{equation*}
\big(\Lambda_{l+1,t-l}+\Omega_{t+1}\big)\bigcap\big(\mathbb{M}_2^{t+1}\times\Set{0}\big)=\emptyset.
\end{equation*}
Hence H\"ormander's sufficient criterion is satisfied.
\end{proof}
\noindent
\textbf{Conclusion.}\,\, We have shown that the renormalization of the time-ordered products and of the anti-chronological products of interactions does not increase the scaling degree estimates for the piecewise renormalized components of the currents which makes them super-renormalizable.

\section*{Acknowledgements}

The author is deeply grateful to Dorothea Bahns for the introduction to this fascinating subject, for providing many mathematical and conceptual tools and for fundamental support during several discussions. Thanks also to Daniela Cadamuro and Markus Fr\"ob for pointing out technical difficulties and to Arne Hofmann for useful suggestions. This work was funded by the German Science Foundation via the Research Training Group 2491 ``Fourier Analysis and Spectral Theory''.

\newpage

\appendix
\section{Proof of Proposition \ref{prop:rec_formula}}
\label{app:Back_coupling_const}
\begin{proof}
First we substitute (\ref{eq:powA}) in (\ref{eq:MBTa}) and use the power series expansion of sine to get
\begin{equation*}
\sum_{\nu=0}^\infty A_{\nu,\xi}\alpha^\nu=-\varphi_\xi+\frac{2}{\alpha}\sum_{\mu=0}^\infty\frac{(-1)^\mu}{(2\mu+1)!}\big(\frac{1}{2}a\big)^{2\mu+1}\Big(\sum_{\nu=0}^\infty A_\nu\alpha^\nu-\varphi\Big)^{2\mu+1}.
\end{equation*}
Requiring that the limit for $\alpha\rightarrow 0$ of this equation exists gives $A_0=\varphi$, hence formula above becomes
\begin{equation}
\label{eq:recMBT}
\sum_{\nu=0}^\infty A_{\nu,\xi}\alpha^\nu=-\varphi_\xi+2\sum_{\mu=0}^\infty\frac{(-1)^\mu}{(2\mu+1)!}\big(\frac{1}{2}a\big)^{2\mu+1}\alpha^{2\mu}\Big(\sum_{\nu=0}^\infty A_{\nu+1}\alpha^\nu\Big)^{2\mu+1}.
\end{equation}
Now we start comparing the coefficients from the left hand side and the right hand side of equation (\ref{eq:recMBT}) for the first orders:
\begin{itemize}
\item At order $0$, we have: $\quad A_{0,\xi}=-\varphi_\xi+2\cdot\frac{1}{2}aA_1\quad\longrightarrow\quad A_1=\frac{2}{a}\varphi_\xi$.
\item At order $1$ we get: $\quad A_{1,\xi}=2\cdot\frac{1}{2}aA_2\quad\longrightarrow\quad A_2=\frac{2}{a^2}\varphi_{\xi\xi}$.
\end{itemize}
For orders $\ge 2$, we rearrange the summation on the right hand side of equation (\ref{eq:recMBT}) in the following way:
\begin{equation*}
\begin{split}
&\sum_{\mu=0}^\infty\frac{(-1)^\mu}{(2\mu+1)!}\big(\frac{1}{2}a\big)^{2\mu+1}\alpha^{2\mu}\Big(\sum_{\nu=0}^\infty A_{\nu+1}\alpha^\nu\Big)^{2\mu+1}= \\
&=\sum_{\mu,\rho=0}^\infty(-1)^\mu\big(\frac{1}{2}a\big)^{2\mu+1}\Big(\sum_{\substack{n_0,\dots,n_\rho\ge 0 \\ n_0+\dots+n_\rho=2\mu+1 \\ 1\cdot n_1+\dots+\rho\cdot n_\rho=\rho}}\frac{A_1^{n_0}\cdots A_{\rho+1}^{n_\rho}}{n_0!\cdots n_\rho!}\Big)\alpha^{\rho+2\mu}.
\end{split}
\end{equation*}
We rewrite the double summation using indexes $\nu:=\rho+2\mu$ and $\beta:=\mu$, so to get the expression:
\begin{equation*}
\sum_{\nu=0}^\infty\alpha^\nu\Big(\sum_{\beta=0}^{[\frac{\nu}{2}]}(-1)^\beta\big(\frac{1}{2}a\big)^{2\beta+1}\sum_{\substack{n_0,\dots,n_{\nu-2\beta}\ge 0 \\ n_0+\dots+n_{\nu-2\beta}=2\beta+1 \\ 1\cdot n_1+\dots +(\nu-2\beta)\cdot n_{\nu-2\beta}=\nu-2\beta}}\frac{A_1^{n_0}\cdots A_{\nu-2\beta+1}^{n_{\nu-2\beta}}}{n_0!\cdots n_{\nu-2\beta}!}\Big),
\end{equation*}
where $[\frac{\nu}{2}]$ is the integer part of $\frac{\nu}{2}$. \\
We now observe that we can decompose the coefficient of $\alpha^\nu$ in two parts, one corresponding to $\beta\ge 1$ and the other for $\beta=0$, respectively:
\begin{equation*}
\begin{split}
&\sum_{\beta=0}^{[\frac{\nu}{2}]}(-1)^\beta\big(\frac{1}{2}a\big)^{2\beta+1}\sum_{\substack{n_0,\dots,n_{\nu-2\beta}\ge 0 \\ n_0+\dots+n_{\nu-2\beta}=2\beta+1 \\ 1\cdot n_1+\dots +(\nu-2\beta)\cdot n_{\nu-2\beta}=\nu-2\beta}}\frac{A_1^{n_0}\cdots A_{\nu-2\beta+1}^{n_{\nu-2\beta}}}{n_0!\cdots n_{\nu-2\beta}!}= \\
&=\sum_{\beta=1}^{[\frac{\nu}{2}]}(-1)^\beta\big(\frac{1}{2}a\big)^{2\beta+1}\sum_{\substack{n_0,\dots,n_{\nu-2\beta}\ge 0 \\ n_0+\dots+n_{\nu-2\beta}=2\beta+1 \\ 1\cdot n_1+\dots +(\nu-2\beta)\cdot n_{\nu-2\beta}=\nu-2\beta}}\frac{A_1^{n_0}\cdots A_{\nu-2\beta+1}^{n_{\nu-2\beta}}}{n_0!\cdots n_{\nu-2\beta}!}+ \\
&\quad+\frac{a}{2}\sum_{\substack{n_0,\dots,n_{\nu}\ge 0 \\ n_0+\dots+n_{\nu}=1 \\ 1\cdot n_1+\dots +\nu\cdot n_{\nu}=\nu}}\frac{A_1^{n_0}\cdots A_{\nu+1}^{n_{\nu}}}{n_0!\cdots n_{\nu}!}.
\end{split}
\end{equation*}
In particular the last term reduces to $\frac{a}{2}A_{\nu+1}$. Comparing the coefficients of the power $\alpha^\nu$, for $\nu\ge 2$, from equation (\ref{eq:recMBT}), we get:
\begin{equation*}
\begin{split}
A_{\nu,\xi}=&2\sum_{\beta=1}^{[\frac{\nu}{2}]}(-1)^\beta\big(\frac{1}{2}a\big)^{2\beta+1}\sum_{\substack{n_0,\dots,n_{\nu-2\beta}\ge 0 \\ n_0+\dots+n_{\nu-2\beta}=2\beta+1 \\ 1\cdot n_1+\dots +(\nu-2\beta)\cdot n_{\nu-2\beta}=\nu-2\beta}}\frac{A_1^{n_0}\cdots A_{\nu-2\beta+1}^{n_{\nu-2\beta}}}{n_0!\cdots n_{\nu-2\beta}!}+ \\
&+aA_{\nu+1}.
\end{split}
\end{equation*}
Extracting $A_{\nu+1}$ and rescaling the summation over $\beta$, we conclude.
\end{proof}

\section{Proof of Proposition \ref{prop:comp_cons_curr}}
\label{app:proof_prop_cons_curr}
\begin{proof}
First we introduce some notation. We define:
\begin{equation*}
\begin{split}
&\varphi+\hat{B}_\alpha\varphi=:\sum_{\nu=0}^\infty A_\nu^+\alpha^\nu,\quad\text{where}\quad
\begin{cases}
A_0^+=2\varphi \\
A_\nu^+=A_\nu\qquad\forall\nu\ge 1,
\end{cases} \\
&\varphi-\hat{B}_\alpha\varphi=:\sum_{\nu=0}^\infty A_\nu^-\alpha^\nu,\quad\text{where}\quad
\begin{cases}
A_0^-=0 \\
A_\nu^-=-A_\nu\qquad\forall\nu\ge 1.
\end{cases}
\end{split}
\end{equation*}
We use the power series expansion of cosine and substitute equations above to get the following expressions for the components of the conserved currents:
\begin{equation*}
\begin{split}
s^{(\alpha)}_1&=\sum_{\mu=0}^\infty\frac{(-1)^\mu}{(2\mu)!}\Big(\frac{1}{2}a\Big)^{2\mu}\Big[\big(\sum_{\nu=0}^\infty A_\nu^+\alpha^\nu\big)^{2\mu}+\big(\sum_{\nu=0}^\infty A_\nu^+(-\alpha)^\nu\big)^{2\mu}\Big], \\
s^{(\alpha)}_2&=-\frac{1}{\alpha^2}\sum_{\mu=1}^\infty\frac{(-1)^\mu}{(2\mu)!}\Big(\frac{1}{2}a\Big)^{2\mu}\Big[\big(\sum_{\nu=0}^\infty A_\nu^-\alpha^\nu\big)^{2\mu}+\big(\sum_{\nu=0}^\infty A_\nu^-(-\alpha)^\nu\big)^{2\mu}\Big].
\end{split}
\end{equation*}
We remark that both formulas above are symmetric in $\alpha$, so only even powers will appear. We further manipulate the two components separately. Starting with $s^{(\alpha)}_1$, we expand $\big(\sum_{\nu=0}^\infty A_\nu^+\alpha^\nu\big)^{2\mu}$ and $\big(\sum_{\nu=0}^\infty A_\nu^+(-\alpha)^\nu\big)^{2\mu}$, collect the coefficients of the even powers $\alpha^{2\rho}$ and obtain:
\begin{equation*}
\begin{split}
s^{(\alpha)}_1=\sum_{\rho=0}^\infty\alpha^{2\rho}\Big[2\sum_{\mu=0}^\infty(-1)^\mu\Big(\frac{1}{2}a\Big)^{2\mu}\Big(\sum_{\substack{n_0,\dots,n_{2\rho}\ge 0 \\ n_0+\dots+n_{2\rho}=2\mu \\ 1\cdot n_1+\dots +2\rho\cdot n_{2\rho}=2\rho}}\frac{(A^+_0)^{n_0}\cdots (A^+_{2\rho})^{n_{2\rho}}}{n_0!\cdots n_{2\rho}!}\Big)\Big].
\end{split}
\end{equation*}
We now concentrate on the coefficient of the power $\alpha^{2\rho}$, we call it $s_1^\rho$:
\begin{equation}
\label{eq:coeff_s1}
s_1^\rho=2\sum_{\mu=0}^\infty(-1)^\mu\Big(\frac{1}{2}a\Big)^{2\mu}\Big(\sum_{\substack{n_0,\dots,n_{2\rho}\ge 0 \\ n_0+\dots+n_{2\rho}=2\mu \\ 1\cdot n_1+\dots +2\rho\cdot n_{2\rho}=2\rho}}\frac{(A^+_0)^{n_0}\cdots (A^+_{2\rho})^{n_{2\rho}}}{n_0!\cdots n_{2\rho}!}\Big).
\end{equation}
Specifically, we want to extract the dependence of the powers of $A_0^+$ on $\mu$. Introducing the index $\beta$ to account for the possible values of the exponent $n_0$, we can rewrite formula (\ref{eq:coeff_s1}) in the following manner:
\begin{equation*}
\label{eq:coeff_s1_n0}
\begin{split}
2\sum_{\beta=0}^{2\rho}\sum_{\mu\ge\frac{\beta}{2}}(-1)^\mu\Big(\frac{1}{2}a\Big)^{2\mu}\frac{(A_0^+)^{2\mu-\beta}}{(2\mu-\beta)!}\bigg(\sum_{\substack{n_1,\dots,n_{2\rho}\ge 0 \\ n_1+\dots+n_{2\rho}=\beta \\ 1\cdot n_1+\dots +2\rho\cdot n_{2\rho}=2\rho}}\frac{(A^+_1)^{n_1}\dots(A^+_{2\rho})^{n_{2\rho}}}{n_1!\dots n_{2\rho}!}\bigg).
\end{split}
\end{equation*}  
Then we distinguish the cases when $\beta$ is even or odd. The terms for $\beta$ even can be collected in the expression:
\begin{equation*}
\label{eq:beta_even}
\begin{split}
2\sum_{\beta=0}^{\rho}\sum_{\mu=\beta}^\infty(-1)^\mu\Big(\frac{1}{2}a\Big)^{2\mu}\frac{(A_0^+)^{2(\mu-\beta)}}{\big(2(\mu-\beta)\big)!}\bigg(\sum_{\substack{n_1,\dots,n_{2\rho}\ge 0 \\ n_1+\dots+n_{2\rho}=2\beta \\ 1\cdot n_1+\dots +2\rho\cdot n_{2\rho}=2\rho}}\frac{(A^+_1)^{n_1}\cdots (A^+_{2\rho})^{n_{2\rho}}}{n_1!\cdots n_{2\rho}!}\bigg).
\end{split}
\end{equation*} 
Rescaling the summation over $\mu$ we recognize the power series expansion of $\cos\big(\frac{1}{2}aA_0^+\big)=\cos(a\varphi)$. Hence for $\beta$ even we obtain the coefficient:
\begin{equation*}
\label{eq:s1_cos}
\begin{split}
\cos(a\varphi)\Big[2\sum_{\beta=0}^{\rho}(-1)^\beta\Big(\frac{1}{2}a\Big)^{2\beta}\sum_{\substack{n_1,\dots,n_{2\rho}\ge 0 \\ n_1+\dots+n_{2\rho}=2\beta \\ 1\cdot n_1+\dots +2\rho\cdot n_{2\rho}=2\rho}}\frac{(A^+_1)^{n_1}\cdots (A^+_{2\rho})^{n_{2\rho}}}{n_1!\cdots n_{2\rho}!}\Big].
\end{split}
\end{equation*}
On the other hand, assuming $\rho\ge 1$, the terms for $\beta$ odd are
\begin{equation*}
\label{eq:beta_odd}
\begin{split}
2\sum_{\beta=0}^{\rho-1}\sum_{\mu=\beta+1}^\infty(-1)^\mu\Big(\frac{1}{2}a\Big)^{2\mu}\frac{(A_0^+)^{2\mu-2\beta-1}}{(2\mu-2\beta-1)!}\sum_{\substack{n_1,\dots,n_{2\rho}\ge 0 \\ n_1+\dots+n_{2\rho}=2\beta+1 \\ 1\cdot n_1+\dots +2\rho\cdot n_{2\rho}=2\rho}}\frac{(A^+_1)^{n_1}\cdots (A^+_{2\rho})^{n_{2\rho}}}{n_1!\cdots n_{2\rho}!}.
\end{split}
\end{equation*}
Rescaling the summation over $\mu$, we recognize the power series expansion of $\sin\big(\frac{1}{2}aA_0^+\big)=\sin(a\varphi)$. Hence for $\beta$ odd we obtain the coefficient:
\begin{equation*}
\label{eq:s1_sin}
\begin{split}
\sin(a\varphi)\Big[2\sum_{\beta=0}^{\rho-1}(-1)^{\beta+1}\Big(\frac{1}{2}a\Big)^{2\beta+1}\sum_{\substack{n_1,\dots,n_{2\rho}\ge 0 \\ n_1+\dots+n_{2\rho}=2\beta+1 \\ 1\cdot n_1+\dots +2\rho\cdot n_{2\rho}=2\rho}}\frac{(A^+_1)^{n_1}\cdots (A^+_{2\rho})^{n_{2\rho}}}{n_1!\cdots n_{2\rho}!}\Big].
\end{split}
\end{equation*}
Using the fact that $A^+_\nu=A_\nu$, for $\nu\ge 1$, and changing the name of the upper index $\rho$ to $N$, we obtain the expected result for $s_1^N$.

For what concerns $s^{(\alpha)}_2$, we use the fact that $A^-_0=0$ to extract a power $\alpha^{2\mu}$, then we divide by $\alpha^2$ and finally rewrite the summations rescaling the indexes, thus obtaining:
\begin{equation*}
\begin{split}
s^{(\alpha)}_2&=\sum_{\mu=0}^\infty\frac{(-1)^\mu}{(2(\mu+1))!}\Big(\frac{1}{2}a\Big)^{2(\mu+1)}\alpha^{2\mu}\cdot \\
&\qquad\qquad\qquad\cdot\Big[\big(\sum_{\nu=0}^\infty A_{\nu+1}^-\alpha^\nu\big)^{2(\mu+1)}+\big(\sum_{\nu=0}^\infty A_{\nu+1}^-(-\alpha)^\nu\big)^{2(\mu+1)}\Big].
\end{split}
\end{equation*}

Expanding $\big(\sum_{\nu=0}^\infty A_{\nu+1}^-\alpha^\nu\big)^{2(\mu+1)}$ and $\big(\sum_{\nu=0}^\infty A_{\nu+1}^-(-\alpha)^\nu\big)^{2(\mu+1)}$ we see that again only the even powers of $\alpha$ survive and they give:
\begin{equation*}
\begin{split}
s^{(\alpha)}_2=\sum_{\mu=0}^\infty &\frac{(-1)^\mu}{(2(\mu+1))!}\Big(\frac{1}{2}a\Big)^{2(\mu+1)}\alpha^{2\mu}\cdot \\
&\cdot\Big[2\sum_{\rho=0}^\infty\alpha^{2\rho}\Big(\sum_{\substack{n_0,\dots,n_{2\rho}\ge 0 \\ n_0+\dots+n_{2\rho}=2(\mu+1) \\ 1\cdot n_1+\dots+ 2\rho\cdot n_{2\rho}=2\rho}}\frac{(2(\mu+1))!}{n_0!\dots n_{2\rho}!}(A^-_1)^{n_0}\dots (A^-_{2\rho+1})^{n_{2\rho}}\Big)\Big].
\end{split}
\end{equation*}

Collecting the powers of $\alpha$, rewriting the summation using indexes $N:=\mu+\rho$ and $\mu$ and recalling that $A_\nu^-=-A_\nu$ for $\nu\ge 1$, we finally obtain that the coefficient of $\alpha^{2N}$ is:
\begin{equation*}
\begin{split}
s_2^N=2\sum_{\mu=0}^N(-1)^\mu\Big(\frac{1}{2}a\Big)^{2(\mu+1)}\sum_{\substack{n_0,\dots,n_{2(N-\mu)}\ge 0 \\ n_0+\dots+n_{2(N-\mu)}=2(\mu+1) \\ 1\cdot n_1+\dots +2(N-\mu)\cdot n_{2(N-\mu)}=2(N-\mu)}}\frac{A_1^{n_0}\cdots A_{2(N-\mu)+1}^{n_{2(N-\mu)}}}{n_0!\cdots n_{2(N-\mu)}!}.
\end{split}
\end{equation*}
\end{proof}

%\clearpage
%\phantomsection


\addcontentsline{toc}{section}{\refname}
\begin{thebibliography}{99}
\bibitem{BR18} D. Bahns, K. Rejzner, \textit{The Quantum Sine-Gordon Model in Perturbative AQFT}, Commun. Math. Phys. 357 (2018) 421.
\bibitem{BRF21} D. Bahns, K. Rejzner and K. Fredenhagen, \textit{Local Nets of Von Neumann Algebras in the Sine-Gordon Model}, Commun. Math. Phys. 383 (2021) 1.
\bibitem{BW} D. Bahns, M. Wrochna, \textit{On-shell Extension of Distributions}, Ann. Henri Poincar\'e 15, 2045–2067 (2014).
\bibitem{BFK} R. Brunetti, K. Fredenhagen, M. K\"ohler, \textit{The microlocal spectrum condition and Wick polynomials on curved spacetimes}, Commun. Math. Phys. 180, 633-652 (1996).
\bibitem{BF} R. Brunetti, K. Fredenhagen, \textit{Microlocal Analysis and Interacting Quantum Field Theories: Renormalization on physical backgrounds}, Commun. Math. Phys. 208, 623-661 (2000).
\bibitem{CadFrob} M. Fr\"ob, D. Cadamuro, \textit{Local operators in the Sine-Gordon model: $\partial_\mu\varphi\,\partial_\nu\varphi$ and the stress tensor}, arXiv:2205.09223 [math-ph] (2022).
\bibitem{CadFrob2} M. Fr\"ob, D. Cadamuro, \textit{A quantum energy inequality in the sine–gordon model}, DOI
10.48550/ARXIV.2212.07377, URL https://arxiv.org/abs/2212.07377 (2022).
%\bibitem{Candu} C. Candu, \textit{Introduction to Integrability}, Lecture notes, %https://edu.itp.phys.ethz.ch/fs13/int/CFT.pdf
\bibitem{DoiFin} A. Doikou and I. Findlay, \textit{Solitons: Conservation laws and dressing methods}, International Journal of Modern Physics A (2019). 
\bibitem{D} M. D\"utsch, \textit{From Classical Field Theory to Perturbative Quantum Field Theory}, Prog. Math. Phys. 74, Birkh\"auser (2019).
\bibitem{DF} M. D\"utsch, K. Fredenhagen, \textit{Causal perturbation theory in terms of retarded products, and a proof of the Action Ward Identity}, Rev. Math. Phys. \textbf{16} (2004), 1291-1348.
\bibitem{Frohlich} J. Fr\"ohlich, \textit{Classical and quantum statistical mechanics in one and two dimensions: Two-component Yukawa \-- and Coulomb systems}, Communications in Mathematical Physics \textbf{47} (1976), no. 3, 233-268.
\bibitem{HollWald01} S. Hollands, R. M. Wald, \textit{Local Wick polynomials and time ordered products of quantum fields in curved spacetime}, Commun. Math. Phys. 223, 289 (2001).
\bibitem{Horm} L. H\"ormander, \textit{The Analysis of Linear Partial Differential Operators I}, 2nd Edition, Springer-Verlag (1990).
\bibitem{R} K. Rejzner, \textit{Perturbative Algebraic Quantum Field Theory. An introduction for Mathematicians}, Mathematical Physics Studies, Springer, 2016.
\bibitem{Shn} Y. Shnir, \textit{Topological and Non-Topological Solitons in Scalar Field Theories}, Cambridge Monographs on Mathematical Physics (2018), Cambridge: Cambridge University Press, doi:10.1017/9781108555623.
\bibitem{Stein} O. Steinmann, \textit{Perturbation Expansions in Axiomatic Field Theory}, Lect. Notes in Phys. 11, Berlin, Springer-Verlag, 1971.
\bibitem{Steu74} H. Steudel, \textit{A Continuum of Conservation Laws for the sine-Gordon and for the Korteweg-de-Vries Equation}, Physics Letters A, Volume 50, Issue 2, 1974. 
\bibitem{Steu76} H. Steudel, \textit{Noether's Theorem and Higher Conservation Laws in Ultrashort Pulse Propagation}, Annalen der Physik. 7. Folge, Band 32, Heft 3, 1976, S. 205-216, J. A. Berth, Leipzig.
\end{thebibliography}
\end{document}